\PassOptionsToPackage{unicode}{hyperref}
\PassOptionsToPackage{hyphens}{url}
\PassOptionsToPackage{dvipsnames,svgnames,x11names}{xcolor}
\documentclass[
  12pt,
]{article}

\usepackage{amsmath,amssymb}
\usepackage{setspace}
\usepackage{iftex}
\ifPDFTeX
  \usepackage[T1]{fontenc}
  \usepackage[utf8]{inputenc}
  \usepackage{textcomp} 
\else 
  \usepackage{unicode-math}
  \defaultfontfeatures{Scale=MatchLowercase}
  \defaultfontfeatures[\rmfamily]{Ligatures=TeX,Scale=1}
\fi
\usepackage{lmodern}
\ifPDFTeX\else  
\fi
\IfFileExists{upquote.sty}{\usepackage{upquote}}{}
\IfFileExists{microtype.sty}{
  \usepackage[]{microtype}
  \UseMicrotypeSet[protrusion]{basicmath} 
}{}
\makeatletter
\@ifundefined{KOMAClassName}{
  \IfFileExists{parskip.sty}{%
    \usepackage{parskip}
  }{
    \setlength{\parindent}{0pt}
    \setlength{\parskip}{6pt plus 2pt minus 1pt}}
}{
  \KOMAoptions{parskip=half}}
\makeatother
\usepackage{xcolor}
\makeatletter
\ifx\paragraph\undefined\else
  \let\oldparagraph\paragraph
  \renewcommand{\paragraph}{
    \@ifstar
      \xxxParagraphStar
      \xxxParagraphNoStar
  }
  \newcommand{\xxxParagraphStar}[1]{\oldparagraph*{#1}\mbox{}}
  \newcommand{\xxxParagraphNoStar}[1]{\oldparagraph{#1}\mbox{}}
\fi
\ifx\subparagraph\undefined\else
  \let\oldsubparagraph\subparagraph
  \renewcommand{\subparagraph}{
    \@ifstar
      \xxxSubParagraphStar
      \xxxSubParagraphNoStar
  }
  \newcommand{\xxxSubParagraphStar}[1]{\oldsubparagraph*{#1}\mbox{}}
  \newcommand{\xxxSubParagraphNoStar}[1]{\oldsubparagraph{#1}\mbox{}}
\fi
\makeatother

\providecommand{\tightlist}{%
  \setlength{\itemsep}{0pt}\setlength{\parskip}{0pt}}\usepackage{longtable,booktabs,array}
\usepackage{calc} 
\usepackage{etoolbox}
\makeatletter
\patchcmd\longtable{\par}{\if@noskipsec\mbox{}\fi\par}{}{}
\makeatother
\IfFileExists{footnotehyper.sty}{\usepackage{footnotehyper}}{\usepackage{footnote}}
\makesavenoteenv{longtable}
\usepackage{graphicx}
\makeatletter
\def\maxwidth{\ifdim\Gin@nat@width>\linewidth\linewidth\else\Gin@nat@width\fi}
\def\maxheight{\ifdim\Gin@nat@height>\textheight\textheight\else\Gin@nat@height\fi}
\makeatother
\setkeys{Gin}{width=\maxwidth,height=\maxheight,keepaspectratio}
\makeatletter
\def\fps@figure{htbp}
\makeatother
\NewDocumentCommand\citeproctext{}{}

\makeatletter
 \let\@cite@ofmt\@firstofone
 \def\@biblabel#1{}
 \def\@cite#1#2{{#1\if@tempswa , #2\fi}}
\makeatother
\newlength{\cslhangindent}
\newlength{\csllabelwidth}
\newenvironment{CSLReferences}[2] 
 {\begin{list}{}{%
  \setlength{\itemindent}{0pt}
  \setlength{\leftmargin}{0pt}
  \setlength{\parsep}{0pt}
  \ifodd #1
   \setlength{\leftmargin}{\cslhangindent}
   \setlength{\itemindent}{-1\cslhangindent}
  \fi
  \setlength{\itemsep}{#2\baselineskip}}}
 {\end{list}}
\usepackage{calc}

\usepackage{accents,setspace,graphicx,srcltx,enumitem}
\usepackage[noblocks]{authblk}

\usepackage{amsmath}
\usepackage{amssymb}
\usepackage{amsfonts}
\usepackage{amsthm}
\usepackage{amsxtra}
\usepackage{amstext}
\usepackage{subcaption}
\usepackage{bbm}
\usepackage[english]{babel}
\usepackage{hyperref}
\usepackage{pgfplots}
\usepackage{booktabs}
\usepgfplotslibrary{fillbetween}
\usetikzlibrary{patterns}

 \usepackage{MnSymbol}  
 \usepackage{adjustbox}
 \usepackage{caption} 
 \newtheorem{assumption}{Assumption}[section]
 
 \theoremstyle{definition}

 \theoremstyle{remark}
 
 \DeclareMathOperator*{\argmin}{arg\,min}

 \allowdisplaybreaks

      \newcounter{tmp}
      \addtocounter{tmp}{1}

       \newcounter{tmp2}
      \addtocounter{tmp2}{1}

         \newcounter{tmp3}
      \addtocounter{tmp3}{1}

\newcommand{\R}{\mathbb{R}}

\usepackage{booktabs}
\usepackage{longtable}
\usepackage{array}
\usepackage{multirow}
\usepackage{wrapfig}
\usepackage{float}
\usepackage{colortbl}
\usepackage{pdflscape}
\usepackage{tabu}
\usepackage{threeparttable}
\usepackage{threeparttablex}
\usepackage[normalem]{ulem}
\usepackage{makecell}
\usepackage{xcolor}
\makeatletter
\@ifpackageloaded{caption}{}{\usepackage{caption}}
\AtBeginDocument{%
\ifdefined\contentsname
  \renewcommand*\contentsname{Table of contents}
\else
  \newcommand\contentsname{Table of contents}
\fi
\ifdefined\listfigurename
  \renewcommand*\listfigurename{List of Figures}
\else
  \newcommand\listfigurename{List of Figures}
\fi
\ifdefined\listtablename
  \renewcommand*\listtablename{List of Tables}
\else
  \newcommand\listtablename{List of Tables}
\fi
\ifdefined\figurename
  \renewcommand*\figurename{Figure}
\else
  \newcommand\figurename{Figure}
\fi
\ifdefined\tablename
  \renewcommand*\tablename{Table}
\else
  \newcommand\tablename{Table}
\fi
}
\@ifpackageloaded{float}{}{\usepackage{float}}
\floatstyle{ruled}
\@ifundefined{c@chapter}{\newfloat{codelisting}{h}{lop}}{\newfloat{codelisting}{h}{lop}[chapter]}
\floatname{codelisting}{Listing}

\usepackage{amsthm}
\theoremstyle{definition}
\newtheorem{example}{Example}[section]
\theoremstyle{plain}
\newtheorem{lemma}{Lemma}[section]
\theoremstyle{plain}
\newtheorem{proposition}{Proposition}[section]
\theoremstyle{plain}
\newtheorem{corollary}{Corollary}[section]
\theoremstyle{plain}
\newtheorem{theorem}{Theorem}[section]
\theoremstyle{remark}
\AtBeginDocument{}
\newtheorem*{remark}{Remark}

\newtheorem{refremark}{Remark}[section]

\makeatother
\makeatletter
\@ifpackageloaded{caption}{}{\usepackage{caption}}
\@ifpackageloaded{subcaption}{}{\usepackage{subcaption}}
\makeatother

\ifLuaTeX
  \usepackage{selnolig}  
\fi
\usepackage{bookmark}

\IfFileExists{xurl.sty}{\usepackage{xurl}}{} 
\hypersetup{
  pdftitle={Identification and Estimation of Consumers' Preferences from Repeated Observations under Nonlinear Pricing},
  pdfauthor={Samuele Centorrino; Frédérique Fève; Jean-Pierre Florens},
  colorlinks=true,
  linkcolor={blue},
  filecolor={Maroon},
  citecolor={Blue},
  urlcolor={Blue},
  pdfcreator={LaTeX via pandoc}}

\title{Identification and Estimation of Consumers' Preferences from
Repeated Observations under Nonlinear Pricing\thanks{\(^\dagger\)Bates
White LLC and Toulouse School of Economics.\newline \(^\ddagger\)Toulouse School of
Economics, University of Toulouse Capitole.\newline Authors wish to thank Ivana
Komunjer, David Martimort, Stefan Pollinger, participants to the 2025
DMV Econometrics Workshop for their useful comments and remarks.
Frédérique Fève and Jean-Pierre Florens acknowledge funding from the
French National Research Agency (ANR) under the Investments for the
Future program (Investissements d'Avenir, grant ANR-17-EURE-0010).}}
\author[1, 2]{Samuele
Centorrino\thanks{Email: \href{mailto:scentorrino@proton.me}{scentorrino@proton.me}}}
\author[2, 3]{Frédérique
Fève\thanks{Email: \href{mailto:frederique.feve@tse-fr.eu}{frederique.feve@tse-fr.eu}}}
\author[2, 3]{Jean-Pierre
Florens\thanks{Email: \href{mailto:jean-pierre.florens@tse-fr.eu}{jean-pierre.florens@tse-fr.eu}}}
\affil[1]{Bates White LLC}
\affil[2]{Toulouse School of Economics}
\affil[3]{University of Toulouse Capitole}
\date{2026-04-16}
\begin{document}
\maketitle
\begin{abstract}
We develop a nonparametric approach to identify and estimate consumer
preferences and unobserved heterogeneity under nonlinear price
schedules. Leveraging variation across multiple price schedules, we show
that both the utility function and the distribution of preference types
can be nonparametrically identified. The quantile function of unobserved
types becomes solution of a functional equation, and we derive
conditions ensuring identification. We propose an iterative approach for
estimation, in which the regularization bias decays exponentially in the
number of iterations while the variance grows only polynomially,
yielding a near-parametric convergence rate. We propose a valid
bootstrap procedure for finite-sample inference and extend the framework
to accommodate potential endogeneity of prices and additional observed
heterogeneity. Monte Carlo simulations and an empirical application to
data from a European mail carrier demonstrate how we can recover the
utility functions and preference distributions in finite samples.
\newline \noindent \textbf{JEL Codes}: C14, D12.
\newline \noindent \textbf{Keywords}: Consumer demand, elasticities,
nonparametric estimation, iterative functional equations.
\end{abstract}

\setstretch{1.25}
\section{Introduction}\label{introduction}

In many economic settings, consumers face nonlinear pricing schedules
where the unit price depends on the quantity purchased. Examples include
public utilities such as electricity, water, and gas, which frequently
employ increasing block pricing, as well as mobile phone plans,
subscription services, and tax schedules with distinct marginal rates.
Understanding and estimating consumer preferences under such nonlinear
budget constraints is crucial for welfare analysis and policy design,
yet it presents significant econometric challenges.

We begin with a model where individuals consume a good with a nonlinear
price schedule. Each consumer has a utility function consisting of a
common component and an idiosyncratic component representing unobserved
preference heterogeneity (Ekeland, Heckman, and Nesheim (2002),
-Ekeland, Heckman, and Nesheim (2004)). When making consumption choices,
consumers maximize utility subject to the nonlinear price schedule,
leading to a first-order condition that relates marginal utility,
marginal price, and individual preference type.

Using this first-order condition, we derive a relationship between the
observable distribution of consumption quantities and the unobservable
distribution of preference types. With a single price schedule, this
relationship does not uniquely identify the utility function and
preference distribution. However, with two different price schedules, we
obtain two constraining equations that can be combined to identify both
components.

Our identification approach involves expressing the problem in terms of
the quantile function of unobserved preference heterogeneity. By
considering how percentiles of the consumption distribution under one
price schedule map to percentiles under another price schedule, we
derive a functional equation that the quantile function must satisfy.
This equation relates the quantile function to the difference in
marginal prices between the two price schedules.

From a technical perspective, our approach formulates the identification
problem as an iterative functional equation. We show that the quantile
function of the unobserved preference component,
\(\Lambda: [0,1]  \rightarrow \mathcal{E} \subset \mathbb{R}\), can be
expressed as the solution of \[
\Lambda(\beta(\alpha)) - \Lambda(\alpha) = r(\alpha),
\] where \(\alpha \in [0,1]\), \(\beta : [0,1]  \rightarrow [0,1]\), and
\(r: [0,1] \rightarrow \mathbb{R}\) are functions that are either known
or easily estimated from the data. This equation is a linear functional
equation in \(\Lambda\). Its solution can be obtained iteratively by
noticing that the equality should hold for any arbitrary value in
\((0,1)\). If \(\beta(\alpha) \rightarrow 0\) as
\(\alpha \rightarrow 0\), and under additional technical conditions that
we detail below, the iteration scheme converges.

We show that this functional equation has a unique solution under
appropriate conditions. The solution to the functional equation can be
expressed as an infinite sum, which forms the basis for our estimation
approach. Once the quantile function is identified, the utility function
can be identified up to a location normalization.

For estimation, we use a regularized estimator that truncates the
infinite sum at a finite value. We establish the statistical properties
of this estimator, analyzing both the bias due to regularization and the
variance due to sampling error. We show the rate of convergence with
appropriate choices of the regularization parameter.

Our work builds upon and extends several strands of literature. The
econometric analysis of consumer behavior under nonlinear budget
constraints dates back to the work of Hausman (1985) on labor supply
under progressive taxation, and the systematic treatment of
piecewise-linear budget constraints in Moffitt (1986). These early
contributions relied on parametric specifications for both preferences
and the distribution of unobserved heterogeneity, typically estimating
the model via maximum likelihood. Blomquist and Newey (2002) relaxed
some of these restrictions by developing a nonparametric estimator, but
still required specific assumptions about the utility function. More
recently, Matzkin (2003) and Blundell, Browning, and Crawford (2008)
have advanced nonparametric approaches to demand estimation, but
primarily in settings with linear budget constraints.

A separate literature has studied demand estimation and preference
recovery under specific forms of nonlinear pricing. Reiss and White
(2005) develop a structural model of residential electricity demand
under multi-tier pricing and show that accounting for the nonlinearity
of the tariff is essential for welfare analysis. In the context of
taxation, Saez (2010) proposes a bunching estimator that uses the excess
mass at kink points of piecewise-linear schedules to recover local
elasticities. Blomquist et al. (2021) show that bunching alone cannot
identify the taxable income elasticity without parametric restrictions
on heterogeneity, but that variation across different budget sets can
achieve identification when the distribution of preferences is
unrestricted and stable. Our approach is closest in spirit to this last
insight: we exploit variation across two price schedules to
nonparametrically identify both the utility function and the full
preference distribution, not just local elasticities. However, we
require the price schedule to be continuously differentiable.

Linear functional equations have been used in other settings in
economics and econometrics (see Carrasco, Florens, and Renault (2007),
for a review). However, to the best of our knowledge, our analysis of a
linear (or nonlinear) iterative functional equation is new. In a
technical annex, we consider the properties of the linear operator which
defines the inverse problem, and show that it is compact but does not
have a bounded inverse. The inverse problem is therefore ill-posed, and
our iterative solution regularizes the inverse by truncation (see
Centorrino, Fève, and Florens (2017), Florens, Racine, and Centorrino
(2018), and Centorrino, Fève, and Florens (2025) for other iterative
methods in the context of ill-posed inverse problems).

We make several contributions to the literature on consumer demand
estimation with nonlinear budget sets. First, we provide a novel
identification strategy that relies on multiple price schedules but
eschews parametric assumptions about preferences or their distribution.
Second, we frame the identification problem as an iterative functional
equation and develop a closed-form solution (see Kuczma, Choczewski, and
Ger (1990)). Third, we propose a fully nonparametric estimator, analyze
its theoretical properties, and propose a valid bootstrap procedure for
finite-sample inference. Fourth, we extend our framework to accommodate
endogenous prices and additional covariates, making it applicable to a
wide range of empirical settings.

We extend our framework to accommodate settings where prices are
potentially endogenous and need to be estimated from data. This involves
an instrumental variable approach, where consumer characteristics that
affect utility but not the price schedule serve as instruments. We also
extend the model to incorporate additional observed heterogeneity,
allowing the utility function to depend on consumer characteristics.

Our methodology provides a flexible framework for analyzing consumer
behavior under nonlinear pricing that relies on minimal assumptions
about functional forms. This approach is particularly valuable in
contexts where imposing parametric restrictions is undesirable due to
concerns about misspecification. The identification strategy leverages
variation in price schedules, making it applicable in many empirical
settings where different price schedules are observed (e.g., different
markets, time periods, or regulatory regimes).

The remainder of the paper is organized as follows.
Section~\ref{sec-model} introduces the model and
Section~\ref{sec-identification} establishes the identification results
using two price schedules. Section~\ref{sec-estimation} develops the
estimation procedure and establishes its statistical properties.
Section~\ref{sec-bootstrap} presents a bootstrap approach for inference,
Section~\ref{sec-addexo} extends the framework to incorporate additional
observed heterogeneity, and Section~\ref{sec-priceest} addresses the
case of endogenous prices. Section~\ref{sec-montecarlo} contains a study
of the small-sample properties of our estimator. Finally,
Section~\ref{sec-application} presents an empirical application to
unaddressed advertising in a European country, where, among other
things, we discuss counterfactual elasticities to nonlinear prices.

\section{Model}\label{sec-model}

We have individuals who consume a quantity \(Q\) of a good for a price
\(P(Q)\), which is not necessarily linear in \(Q\). Every individual has
a utility function that can be written as \(u(Q) + Q \varepsilon\),
where \(\varepsilon\) is an unobserved individual type generated by a
probability distribution on \(\R\), such that
\(\varepsilon \sim F_\varepsilon\).

We impose the following regularity conditions.

\begin{assumption}~ \label{ass:regcond}
\begin{enumerate}
\item[(i)] $Q \in \mathcal{Q}$, a compact subset of $\R_{\geq 0}$, with $P(Q)$ and $u(Q)$ twice continuously differentiable, and $u^{\prime \prime}(Q) - \tau P^{\prime \prime}(Q) <0$, for $\tau > 0$, and all $Q \in \mathcal{Q}$. 
\item[(ii)]  $\varepsilon \in \mathcal{E} = \left[ \underset{\bar{}}{\varepsilon}, \bar{\varepsilon} \right]$, a compact subset of $\R$. 
\item[(iii)]  $F_\varepsilon: \mathcal{E} \rightarrow [0,1]$ is strictly monotone increasing and continuous, with $\Lambda = F^{-1}_\varepsilon$. Its first derivative $f_\varepsilon$ is uniformly bounded away from $0$ and $\infty$ on $\mathcal{E}$.
\end{enumerate}
\end{assumption}

Every individual solves the following maximization problem
\[\max_{Q} u(Q) + Q \varepsilon - \tau( P(Q)- B),\] where \(\tau > 0\)
is a Lagrange multiplier and \(B\) is the total budget available to
purchase the good. As \(B\) is not observed, we assume that the
constraint is always binding, and therefore \(\tau\) is strictly
positive. For a given \(\varepsilon\), the optimal quantity demanded,
\(Q^\ast\), satisfies the following first order condition
\begin{equation}\phantomsection\label{eq-eq:optquant}{
u^\prime(\bar{Q}) - \tau P^\prime(\bar{Q}) = -\varepsilon, \text{ with } Q^\ast = \bar{Q} \vee 0. 
}\end{equation}

Under the conditions in Assumption \ref{ass:regcond}(i), the
maximization problem in Equation~\ref{eq-eq:optquant} has a unique
solution. Assumptions \ref{ass:regcond}(ii)-(iii) are regularity
conditions on the support and cdf of \(\varepsilon\).

We let \(\varphi(Q) = u^\prime(Q) - \tau P^\prime(Q)\) which, from
Assumption \ref{ass:regcond}, is a strictly decreasing function defined
on \(\R^{+}\).

The price schedule \(P(Q)\) is supposed to be known. Hence, the unknown
parameters of this model are the utility function \(u(Q)\), the
distribution of types \(F_\varepsilon\), and the Lagrange multiplier
\(\tau\).

We observe the quantity consumed by each individual whose cdf
\(G(q) = \Pr(Q\leq q)\), with compact support
\(\left[ \underset{\bar{}}{q}, \bar{q} \right]\). This cdf can have a
positive mass at \(0\), if a positive mass of individuals does not
consume the good.

\begin{refremark}
In many cases, \(P\) is piece-wise linear (e.g., two-part tariffs). We
assume that we can reconstruct an optimal function whose envelope is the
two-part tariff.

\label{rem-rem:piecewise}

\end{refremark}

The objective of this work is to estimate \(u(Q)\), \(F_\varepsilon\),
and \(\tau\), when \((G,P)\) are known or can be identified from the
data. This model is identified when we have multiple observations
available and under further technical conditions.

For any \(G\), if \(u^\prime(Q)\) is shifted by a constant \(c\), and
\(\varepsilon_c = \varepsilon - c\), then the solution of
Equation~\ref{eq-eq:optquant} is not modified. To get rid of this
identification problem, we make the following normalization assumption.

\begin{assumption} \label{ass:errorbounds}
$\underset{\bar{}}{\varepsilon} = 0$. 
\end{assumption}

The fundamental relation between observables and parameters of the model
is then given by

\[
\Pr(Q \leq q) = \Pr(\varphi(Q) \geq \varphi(q)) = \Pr(-\varepsilon \geq \varphi(q)) =\Pr(\varepsilon \leq -\varphi(q)),
\]

which implies

\begin{equation}\phantomsection\label{eq-eq:idequation}{
G(q) = F_\varepsilon( -\varphi(q)), \text{ with } q \geq 0. 
}\end{equation}

We cannot uniquely identify the triplet \((u,F_\varepsilon,\tau)\) from
Equation~\ref{eq-eq:idequation}. One could impose a variety of
restrictions based on the knowledge of either \(F_\varepsilon\) or
\(u\), or other parametric restriction to aid identification. Here, we
follow a purely nonparametric approach.

First, we may have an issue of selection at \(0\), as individuals who do
not consume are usually not observed. In that scenario, the cdf of
\(\varepsilon\) captures the distribution of types conditional on
consuming a positive quantity of the good. We can potentially have two
cases:

\begin{enumerate}
\def\labelenumi{\alph{enumi})}
\tightlist
\item
  \(\varphi(0) \geq 0\). In this case all individuals are consuming, and
  \(\underset{\bar{}}{q}\), the lower bound of the distribution of \(Q\)
  is positive or zero. \(\underset{\bar{}}{q}\) is identified from the
  data, and we can reparametrize the problem by defining
  \(\tilde{u}(x) = u(\underset{\bar{}}{q} + x)\),
  \(\tilde{P}(x) = P(\underset{\bar{}}{q} + x)\),
  \(\tilde{\varphi}(x) = \tilde{u}^\prime(x) - \tau \tilde{P}^\prime(x)\),
  with \(\tilde{\varphi}(0) = 0\), and \(\tilde{G}(x)\) with support
  \(\left[ 0, \bar{q} - \underset{\bar{}}{q}\right]\) (see Figure
  \ref{fig:identification1}).
\item
  \(\varphi(0) < 0\). In this case, there are some individuals who do
  not consume, and \(\Pr(Q = 0)=F_\varepsilon (-\varphi(0))\).
  Therefore, we can simply consider the distribution of \(Q\)
  conditional on \(Q >0\). We would not be able to identify
  \(F_\varepsilon\) between \(0\) and \(-\varphi(0)\). The object of
  interest is thus \(F_{\varepsilon^\ast}\), with
  \(\varepsilon^\ast = \varepsilon + \varphi(0)\), conditional on
  \(\varepsilon^\ast > 0\). Hence, \[
  G(q \vert Q > 0) = \Pr \left( \varepsilon + \varphi(0) \leq -\varphi(q) + \varphi(0) \vert \varepsilon > -\varphi(0) \right) = F_{\varepsilon^\ast} \left( -\tilde{\varphi}(q) \right),
  \] with \(\tilde{\varphi}(q) = \varphi(q) - \varphi(0)\) (see Figure
  \ref{fig:identification2}).
\end{enumerate}

\begin{figure}[ht]
    \centering
  \begin{subfigure}{.45\textwidth}
      \begin{tikzpicture}
          \begin{axis}[
              axis lines=center,
              axis on top=true,
              xlabel=$q$,
              ylabel={$-\varepsilon$},
              ymin=-4.5,
              ymax=3,
              xmin=-3,
              xmax=4,
              xtick=\empty,
              ytick=\empty,
              every axis x label/.style={
                  at={(ticklabel* cs:1.05)},
                  anchor=west,
              },
              every axis y label/.style={
                  at={(ticklabel* cs:1.05)},
                  anchor=south,
              },
          ]
      
          \addplot [domain=0:4, samples=200, smooth, thick] {-0.5*(x-2)*(x+2)};
          \node at (axis cs:1,1) {$\varphi(q)$};
          \node at (axis cs:-0.25,2) {$-\underset{\bar{}}{\varepsilon}$};
          
          
          \addplot [domain=0:3.5, samples=200, smooth, dotted] {-4.125};
          \node at (axis cs:-0.25,-4.125) {$-\bar{\varepsilon}$};
          \draw [dotted] (axis cs:3.5,-4.125) -- (axis cs:3.5,0);

          \addplot [domain=0:2.75, samples=200, smooth, dashed] {-1.78125};
          \node at (axis cs:2.75,0.5) {$Q$};
          \draw [dashed] (axis cs:2.75,-1.78125) -- (axis cs:2.75,0);
          
          \addplot [domain=2:3.5, samples=100, smooth, thick] {(x-2)*((3.5 - x)^4)/0.253125};
          \node at (axis cs:3.25,1.5) {$G^\prime(q)$};
          \node at (axis cs:2,-0.5) {$\underset{\bar{}}{q}$};
          \node at (axis cs:3.5,-0.5) {$\bar{q}$};
          \node at (axis cs:2,0)[circle,fill,inner sep=1.5pt]{};
          \node at (axis cs:3.5,0)[circle,fill,inner sep=1.5pt]{};
          
          \addplot [domain=-4.125:0, samples=200, smooth, thick] ({-abs(x)*((4.125 - abs(x))^4)/39.8105804},{x});
          \node at (axis cs:-0.75,-1.5) {$F^\prime(\varepsilon)$};
      
          \end{axis}
      \end{tikzpicture}
      \caption{Identification when $\varphi(0) \geq 0$.}
        \label{fig:identification1}
      \end{subfigure}
  \begin{subfigure}{.45\textwidth}
      \begin{tikzpicture}
          \begin{axis}[
              axis lines=center,
              axis on top=true,
              xlabel=$q$,
              ylabel={$-\varepsilon$},
              ymin=-7.25,
              ymax=3,
              xmin=-3,
              xmax=4,
              xtick=\empty,
              ytick=\empty,
              every axis x label/.style={
                  at={(ticklabel* cs:1.05)},
                  anchor=west,
              },
              every axis y label/.style={
                  at={(ticklabel* cs:1.05)},
                  anchor=south,
              },
          ]
      
          \addplot [domain=0:4, samples=200, smooth, thick] {-1-(2/3)*x^2};
          \node at (axis cs:0.75,-2.25) {$\varphi(q)$};
          
          
          \addplot [domain=0:3, samples=200, smooth, dotted] {-7};
          \node at (axis cs:-0.25,-7) {$-\bar{\varepsilon}$};
          \draw [dotted] (axis cs:3,-7) -- (axis cs:3,0);
          \draw [name path=A,dotted] (axis cs:-(5^4/259.2,-1) -- (axis cs:0,-1);
          \addplot [name path=B,domain=-1:0, samples=200, smooth, thick, dashed] ({-abs(x)*((6 - abs(x))^4)/259.2},{x});
          \addplot[color=gray!60]fill between[of=A and B];
          
          \addplot [domain=0:3, samples=100, smooth, thick] {x*((3 - x)^4)/8.1};
          \node at (axis cs:2.5,1.25) {$G^\prime(q \vert q > 0 )$};
          \node at (axis cs:3,-0.5) {$\bar{q}$};
          \node at (axis cs:0,0)[circle,fill,inner sep=1.5pt]{};
          \node at (axis cs:3,0)[circle,fill,inner sep=1.5pt]{};
          
          \addplot [domain=-6:0, samples=200, smooth, thick] ({-abs(x)*((6 - abs(x))^4)/259.2},{x});
          \node at (axis cs:-1,0.5) {$F^\prime(\varepsilon)$};
          
          \addplot [domain=-7:-1, samples=200, smooth, thick, dashed] ({-(abs(x)-1)*((7 - abs(x))^4)/259.2},{x});
          \node at (axis cs:-1.5,-5) {$F^\prime(\varepsilon \vert \varepsilon > -\varphi(0))$};
      
          \end{axis}
      \end{tikzpicture}
      \caption{Identification when $\varphi(0) < 0$.}
        \label{fig:identification2}
    \end{subfigure}
\end{figure}

This discussion shows that in both cases we can normalize
\(\varphi(0) = 0\), and the object of interest is the distribution of
\(\varepsilon\) conditional on \(Q>0\). The cdfs \(G\) and
\(F_\varepsilon\) should always be taken conditionally on \(Q > 0\),
but, for simplicity, we omit this in the following.

\section{Identification with two samples}\label{sec-identification}

Let us consider the model defined by Equation~\ref{eq-eq:idequation},
with the additional normalization \(\varphi(0) = 0\). We are going to
assume that we have two repeated cross-sections generated by different
price schedules, \(P_1(Q)\) and \(P_2(Q)\), and distribution of
quantities \(G_1(q)\) and \(G_2(q)\), while \(F_\varepsilon\) and \(u\)
are the same across the two samples. We assume the following.

\begin{assumption} \label{ass:priceder}
$\tau_1 P^\prime_1(0)= \tau_2 P^\prime_2(0) = u^\prime(0)$.
\end{assumption}

Assumption \ref{ass:priceder} requires the marginal utility to be equal
to the derivative of the budget constraints at \(0\) for both price
schedules. This is a boundary condition, consistent with the
normalization \(\varphi_1(0) = \varphi_2(0) = 0\). When
\(\varepsilon = 0\), then \(0\) is the optimal solution to the agent's
maximization problem. While, for \(\varepsilon > 0\), we must observe
\(Q^\ast > 0\).

We also require the following lemma.

\begin{lemma}[]\protect\hypertarget{lem-lem:cdfcont}{}\label{lem-lem:cdfcont}

Let Assumptions \ref{ass:regcond}-\ref{ass:priceder} hold. Then we have
the following

\begin{enumerate}
\def\labelenumi{(\roman{enumi})}
\tightlist
\item
  Let \(G_j = F_\varepsilon  \circ (-\varphi_j)\), and
  \(\varphi_j = u^\prime - \tau_j P^\prime_j\) for \(j = 1,2\). Then,
  for \(j = 1,2\), \(G_j\) is a continuous and strictly monotone cdf
  with support equal to \(\mathcal{Q}_j = \left[0 , \bar{q}_j \right]\).
\item
  \(F_\varepsilon ( \varepsilon \vert Q_1 > 0 ) = F_\varepsilon ( \varepsilon \vert Q_2 > 0 ) =  F_\varepsilon ( \varepsilon )\).
\end{enumerate}

\end{lemma}

By Lemma~\ref{lem-lem:cdfcont},
\(\mathcal{Q} = \mathcal{Q}_1 \cap \mathcal{Q}_2 = \left[ 0 , \min(\bar{q}_1, \bar{q}_2) \right]\),
and we assume wlog that \(\bar{q}_1 \geq \bar{q}_2\). Finally, we impose
the following.

\begin{assumption} \label{ass:cdfintersect}
$G_1(q) = G_2(q)$ only on a finite number of points in $\mathcal{Q}$.
\end{assumption}

Assumption \ref{ass:cdfintersect} imposes that repeated observations are
informative as long as there is an overlap between the quantities
consumed, and identification is possible only in the intersection
between the two supports. If that intersection is empty, or if
\(G_1 \overset{a.s.}{=} G_2\), then we revert to the identification
issue with only one sample. If \(G_1 \overset{a.s.}{=} G_2\) only on a
strict subset of \(\mathcal{Q}\), then one may be able to partially
identify the objects of interest, but we do explicitly treat this case
in the present paper.

We thus have two identifying restrictions

\begin{align*}
\Lambda (G_1(q)) =& -\varphi_1(q)  = -u^\prime(q) + \tau_1 P^\prime_1 (q) \\
\Lambda (G_2(q)) =& -\varphi_2(q)= -u^\prime(q) + \tau_2 P^\prime_2 (q),
\end{align*} where \(\Lambda = F^{-1}_\varepsilon\) is the quantile
function of \(\varepsilon\). Note that by Assumption \ref{ass:priceder}
and the normalization \(\varphi_j(0) = 0\), we have
\(\Lambda(G_j(0)) = \Lambda(0) = 0\).

Taking the difference between these two equations, we obtain
\begin{equation}\phantomsection\label{eq-eq:twosamples}{
\Lambda (G_2(q)) - \Lambda (G_1(q)) = \tau_2 P^\prime_2(q) - \tau_1 P^\prime_1(q)
}\end{equation}

Our goal is to identify the function \(\Lambda\), and the parameters
\((\tau_1,\tau_2)\), from Equation~\ref{eq-eq:twosamples} when the price
schedules are known or can be identified from the data, and for given
\(\lbrace G_j,j = 1,2 \rbrace\).

We first discuss the separate identification of \((\tau_1,\tau_2)\).
Equation~\ref{eq-eq:twosamples} can only identify these parameters up to
scale. For instance, if we define \(\tilde{\tau}_1 = c \tau_1\), for any
constant \(c>0\), then Equation~\ref{eq-eq:twosamples} is equivalent to
\[
c\Lambda (G_2(q)) - c\Lambda (G_1(q)) = c\tau_2 P^\prime_2(q) - \tilde{\tau}_1 P^\prime_1(q).
\]

A convenient normalization is to impose \(\tau_2 = 1\). \(\tau_1\) can
then be identified as follows. Assume there exists \(\tilde{q}\), such
that \(G_2(\tilde{q}) = G_1(\tilde{q})\). Then, assuming that both price
derivatives exist at the point \(\tilde{q}\), we must have

\[
\tau_1 = \frac{P_2^\prime (\tilde{q})}{P_1^\prime (\tilde{q})}.
\]

If either cdf is stochastically dominated by the other one on
\(\mathcal{Q}\), then \(\tilde{q}\) is a point at the boundary of the
support. If the two cdfs intersect in the interior of \(\mathcal{Q}\),
we can use the value of the price derivative at the crossing point to
identify and estimate \(\tau_1\).

All these quantities can be identified from the data, and therefore
\(\tau_1\) is identified under the additional condition that the first
derivative of the price function is finite and not equal to zero at the
chosen crossing point, \(\tilde{q}\). When the price schedule is known,
and because of Assumption \ref{ass:priceder}, then
\(\tau_1 = \frac{P_2^\prime (0)}{P_1^\prime (0)}\), which implies that
\(\tau_1\) is known when the two price functions are known. If the price
functions are unknown but can be estimated at \(\sqrt{n}\) rate, we can
construct a \(\sqrt{n}\)-consistent estimator of \(\tau_1\) based on the
estimation of \(G_1\) and \(G_2\). As this is relatively
straightforward, we consider for simplicity the case where
\(\tau_1 = \tau_2 = 1\), and we show the properties of this estimation
approach in simulations.

Let us now turn to the identification of \(\Lambda\). For any
\(\alpha \in [0,1]\), let
\(\beta(\alpha) = G_2 \left( G_1^{-1}(\alpha)\right)\), where
\(\beta:  [0,1] \rightarrow  [0,1]\), with \(\beta(0) = 0\). This
function is strictly monotone by Lemma~\ref{lem-lem:cdfcont}. Therefore,
we can rewrite equation Equation~\ref{eq-eq:twosamples} as

\begin{equation}\phantomsection\label{eq-eq:functional1}{
\Lambda (\beta(\alpha)) - \Lambda (\alpha) = r(\alpha),
}\end{equation}

where
\(r(\alpha) = P^\prime_2(G_1^{-1}(\alpha)) - P^\prime_1(G_1^{-1}(\alpha))\).
Letting \[
(T\Lambda)(\alpha) = \Lambda (\beta(\alpha)) - \Lambda (\alpha), 
\] Equation~\ref{eq-eq:functional1} becomes

\begin{equation}\phantomsection\label{eq-eq:functional2}{
T\Lambda = r, 
}\end{equation}

where \(T\) is an operator. The functional equation in
Equation~\ref{eq-eq:functional2} is a linear equation (see Kuczma,
Choczewski, and Ger (1990), Ch. 2). Structural identifying restrictions
in several econometric models can be cast as linear functional equations
(for instance, in the case of nonparametric instrumental regressions,
see, e.g. Newey and Powell (2003), Carrasco, Florens, and Renault
(2007), Darolles et al. (2011)). However, to the best of our knowledge,
results about identification and estimation of
Equation~\ref{eq-eq:functional2} are new in econometrics. We are going
to consider Equation~\ref{eq-eq:functional2} as a mapping from
\(\mathcal{L}\) to \(\mathcal{L}\), where
\(\mathcal{L}= \mathcal{C}_{[0,1]}\), the class of continuous uniformly
bounded functions in \([0,1]\).

The operator \(T\) is linear and, to study identification of \(\Lambda\)
in \(\mathcal{L}\), we consider its null space. We first prove the
following.

\begin{proposition}[]\protect\hypertarget{prp-prp:identification}{}\label{prp-prp:identification}

Let Assumptions \ref{ass:regcond}-\ref{ass:cdfintersect} and
Lemma~\ref{lem-lem:cdfcont} hold. Then \(T \Lambda \overset{a.s.}{=} 0\)
if and only if \(\Lambda\) is constant in \([0,1]\).

\end{proposition}

\begin{corollary}[]\protect\hypertarget{cor-cor:identification}{}\label{cor-cor:identification}

Let \(\Lambda(0) =0\). Then \(T \Lambda \overset{a.s.}{=} 0\) implies
\(\Lambda \overset{a.s.}{=} 0\).

\end{corollary}

The proof of this Corollary follows from the proof of
Proposition~\ref{prp-prp:identification}, with the additional constraint
that, when \(\Lambda\) is constant, it is also equal to zero almost
everywhere on \([0,1]\). This implies that \(T\) is injective and that
\(\Lambda\) is identified. The results of
Proposition~\ref{prp-prp:identification}, and
Corollary~\ref{cor-cor:identification} strongly depend on the continuity
of \(\Lambda\) over the support of \(\varepsilon\). In particular, we
rule out the existence of gaps in the distribution of \(\varepsilon\),
i.e., sets where the cdf of \(\varepsilon\) is constant.

When \(\Lambda\) is identified, then the function \(u\) is identified up
to location as

\begin{equation}\phantomsection\label{eq-utility-identification}{
u (q) = \int_{0}^q \left[ P^\prime_i (t) - \Lambda \left( G_i (t) \right) \right] dt, \text{ with } i = 1,2.
}\end{equation}

\begin{remark}[An alternative approach]
Assuming that $\Lambda$ is continuously differentiable, we can rewrite Equation \ref{eq-eq:functional1} as 
$$
\int_{\alpha}^{\beta(\alpha)} \lambda(t) dt = r(\alpha),
$$
where $\lambda(t) = d\Lambda(t)/dt$. $\lambda$ can be obtained as the solution of an integral equation of the first kind, which requires some type of regularization for consistent estimation. Furthermore, regularity assumptions should be made directly on the first derivative. We expect this estimator to also have good properties in finite sample, and we provide a comparison with a Tikhonov-regularized estimator of the first derivative in simulations. We defer a thorough theoretical comparison for future research. 
\end{remark}

\section{Estimation}\label{sec-estimation}

Let us focus on the simple case in which
\(\mathcal{Q} = \mathcal{Q}_1 = \mathcal{Q}_2\), and
\(G_1(q) > G_2(q)\), which implies that \(\beta(\alpha) < \alpha\). The
latter is a stochastic dominance assumption that can be tested (see,
e.g., Linton, Maasoumi, and Whang (2005)). In the case in which
\(\beta(\alpha) > \alpha\), we can simply reverse the roles of \(G_1\)
and \(G_2\). When the functions \(G_1\) and \(G_2\) cross in the
interval \((0,1)\), one can consider estimation between the crossing
points. From the practical perspective, this case does not modify the
estimation strategy, and we illustrate it using simulations in
Section~\ref{sec-montecarlo}. When \(\beta(\alpha) < \alpha\), with
\(\beta\) continuous, and \(\beta(0) = 0\), \(\beta^{(k)}(\alpha)\)
converges to \(0\) as \(k \rightarrow \infty\).

Having established in Section~\ref{sec-identification} that
Equation~\ref{eq-eq:functional1} has a unique solution under our
assumptions, let us consider the closed-form expression of that
solution. First, notice that for \(\beta(\alpha) \in (0,1)\), one needs
to have that

\begin{equation}\phantomsection\label{eq-eq:functional1b}{
\Lambda (\beta^{(2)}(\alpha)) - \Lambda (\beta(\alpha)) = r(\beta(\alpha)).
}\end{equation}

Taking the sum between Equation~\ref{eq-eq:functional1} and
Equation~\ref{eq-eq:functional1b}, we have that

\begin{equation}\phantomsection\label{eq-eq:functionalsol1}{
\Lambda (\beta^{(2)}(\alpha)) - \Lambda (\alpha) = r(\beta(\alpha)) + r(\alpha).
}\end{equation}

More generally, for any \(k \geq 1\), we must have
\begin{equation}\phantomsection\label{eq-eq:functional1k}{
\Lambda (\beta^{(k)}(\alpha)) - \Lambda (\beta^{(k-1)}(\alpha)) = r(\beta^{(k-1)}(\alpha)).
}\end{equation}

Iterating the reasoning in Equation~\ref{eq-eq:functionalsol1}, we have
that a general solution to Equation~\ref{eq-eq:functional1} can be
written as \begin{equation}\phantomsection\label{eq-eq:sollambda}{
\Lambda^\ast(\alpha) = - \sum_{k = 0}^\infty r\left( \beta^{(k)} (\alpha) \right),
}\end{equation} (see Kuczma, Choczewski, and Ger (1990), Th. 2.3.5,
p.~65), and this solution exists if the sequence on the right-hand-side
is convergent for all \(\alpha \in [0,1)\).

The following two assumptions provide sufficient conditions for
Equation~\ref{eq-eq:sollambda} to converge.

\begin{assumption} \label{ass:attractive}
$0$ is a strongly attractive point. That is, there exist constants $(\delta,\theta) \in (0,1)$, such that
$$
\beta(x) < \theta x, \text{ for } x \in (0,\delta).
$$
\end{assumption}

\begin{assumption} \label{ass:priceregholder}
There exist positive constants $\lbrace C,\delta_0,\kappa \rbrace$, with $\delta_0 \in (0,1)$, and $\kappa > 0$, such that
$$
\vert r(\alpha) - r(0) \vert  = \vert P^\prime_1 \left( G_1^{-1} (\alpha) \right) - P^\prime_2\left( G_1^{-1} (\alpha) \right) \vert \leq C \alpha^\kappa, \text{ for } \alpha \in (0,\delta_0).
$$
\end{assumption}

Assumption \ref{ass:priceregholder} is a Hölder continuity condition of
the function \(r\) around \(0\).

Assumption \ref{ass:attractive} implies that, \[
\Lambda (\beta^{(k)}(\alpha)) \rightarrow 0, \text{ as } k \rightarrow \infty.
\] Assumptions \ref{ass:priceregholder} and \ref{ass:attractive}
together imply that \[
\vert r(\beta^{(k)}(\alpha)) \vert \leq C \beta^{(k)}(\alpha)^\kappa < C \theta^{k \kappa} \alpha^{\kappa},
\] which implies that Equation~\ref{eq-eq:sollambda} converges for
\(\alpha \in [0,1)\).

In practice, we regularize the solution in
Equation~\ref{eq-eq:sollambda} by choosing a finite number of
iterations, \(N\). That is,

\begin{equation}\phantomsection\label{eq-eq:sollambdareg}{
\Lambda^{\ast,N}(\alpha) = - \sum_{k = 0}^N r\left( \beta^{(k)} (\alpha) \right).
}\end{equation}

The approximation bias can be characterized by the regularity of
\(\Lambda\), which itself depends on the regularity of \(r\) given in
Assumption \ref{ass:priceregholder}. We thus have the following.

\begin{lemma}[Bias]\protect\hypertarget{lem-lembias}{}\label{lem-lembias}

Let Assumptions \ref{ass:attractive} and \ref{ass:priceregholder} hold.
Then \[
\sup_{\alpha \in (0,1)}\vert \Lambda^{\ast,N}(\alpha) - \Lambda^{\ast}(\alpha) \vert = O(\theta^{(N+1)\kappa}).
\]

\end{lemma}

Under our Assumptions, we show that the bias decreases exponentially
with the regularization parameters \(N\), which differentiate our result
from the existing literature on ill-posed inverse problem (Newey and
Powell (2003),Darolles et al. (2011)).

\subsection{Variance}\label{variance}

We have iid observations from samples \(1\) and \(2\), denoted
\(\lbrace \left(P_{ji},Q_{ji}\right),i = 1,\dots,n_j\rbrace\) for
\(j = 1,2\). For simplicity, and without loss of generality, we assume
that \(n_1 = n_2 = n\). As the support of \(Q_1\) and \(Q_2\) is assumed
to be compact, their cdfs, \((G_1,G_2)\), and quantile functions,
\((G^{-1}_1,G_2^{-1})\), can be estimated at a parametric rate (Li and
Racine (2008)). Regarding the price function, we first focus on the case
in which the price schedule (and hence its derivative) are known to both
the decision maker and the econometrician, and we defer the discussion
about their estimation in Section~\ref{sec-priceest}.

We consider the following nonparametric estimators \begin{align*}
\hat{G}_{h,j}(q) =& \frac{1}{n_j} \sum_{i = 1}^{n_j} \bar{K} \left( \frac{Q_{ji} - q}{h_{j}}\right), \\
\hat{G}^{-1}_{h,j}(\alpha) =& \left( \hat{G}_{h,j}\right)^{-1}(\alpha),
\end{align*} where \(\bar{K}\) is the cdf of a nonparametric kernel
function \(K\) and \(h_{j}\) is a bandwidth parameter for \(j = 1,2\).
Let \(\alpha \in (0,1)\) be given. We consider here a kernel estimation
of the cdfs, as we can leverage known results about their uniform
in-bandwidth consistency. However, one could decide to use sieve
estimators of the same quantities. Then \begin{align*}
\hat{\beta}^{(j)}_{h} (\alpha) =& \left[ \hat{G}_{h,2} \hat{G}_{h,1}^{-1} \right]^{(j)} (\alpha)\\
\hat{r}(\hat{\beta}^{(j)}_{h} (\alpha)) =& \left( P^{\prime}_1 - P^{\prime}_2 \right)\left( \hat{G}^{-1}_{h,1} \left( \hat{\beta}^{(j)}_{h} (\alpha) \right) \right),
\end{align*} where \(P^{\prime}\) is the known derivative of the price
function. We take the kernel smoothed estimator, as we can immediately
apply existing results to prove its uniform consistency. However, we
could have as well used sieve estimators for the same purpose (Tsybakov
(2008)).

Our regularized estimator of the solution \(\Lambda^\ast\) is then given
by \begin{equation}\phantomsection\label{eq-eq:sollambdahat}{
\hat{\Lambda}_h^{\ast,N}(\alpha) = - \sum_{k = 0}^N r\left( \hat{\beta}^{(k)}_h (\alpha) \right).
}\end{equation}

When using the empirical cdf and the empirical quantile function for
estimation, standard results from empirical process theory would imply
the uniform convergence of \(\hat{G}_{h,j}(q)\) and
\(\hat{G}^{-1}_{h,j}(\alpha)\), for \(j = 1,2\) over the class
\(\mathcal{L}\). However, as we use a smooth version of the empirical
cdf and quantile function, the properties of our estimator depend on the
choice of the bandwidth parameter. We have the following Proposition.

\begin{proposition}[Mason and Swanepoel (2013), Proposition
1.5]\protect\hypertarget{prp-prp:cdfest}{}\label{prp-prp:cdfest}

Let Lemma~\ref{lem-lem:cdfcont} hold and \(h_{j} \rightarrow 0\) as
\(n \rightarrow \infty\) for \(j = 1,2\). Further assume the following

\begin{itemize}
\tightlist
\item
  The probability density functions of \(Q_1\) and \(Q_2\), \(g_1\) and
  \(g_2\), exist, are continuous on \(\mathcal{Q}\), and almost surely
  bounded away from zero and infinity on their support.
\item
  There exists a sequence \(b_n \in (0,1)\) such that
  \(b_n \rightarrow 0\) with probability \(1\).
\end{itemize}

Then \begin{align*}
\lim_{n \rightarrow \infty} \sup_{0 <h\leq b_n} \sup_{q \in \mathcal{Q}} \vert \hat{G}_{h,j}(q) - G_j(q) \vert = 0 \\
\lim_{n \rightarrow \infty} \sup_{0 <h\leq b_n} \sup_{\alpha \in [0,1]} \vert \hat{G}^{-1}_{h,j} (\alpha)- G^{-1}_j(\alpha) \vert = 0,
\end{align*} and \begin{align*}
\sup_{0 <h\leq b_n} \sqrt{n_j} \left( \hat{G}_{h,j} - G_j\right) \Rightarrow& \mathbb{B} \circ G_j \\
\sup_{0 <h\leq b_n} \sqrt{n_j} \left( \hat{G}^{-1}_{h,j} - G^{-1}_j\right) \Rightarrow& \frac{\mathbb{B}}{g_j \circ G^{-1}_j},
\end{align*} for \(j = 1,2\), where \(\Rightarrow\) denotes weak
convergence in the Skorohod space equipped with the uniform norm, and
\(\mathbb{B}\) is a standard Brownian bridge.

\end{proposition}

Proposition~\ref{prp-prp:cdfest} gives a version of the
Glivenko-Cantelly and Donsker theorems when the cdf is estimated using a
smooth kernel function. This result is uniform in the bandwidth
parameter (see, e.g., Van der Vaart (1998) p.~266 for a textbook
reference). By the law of iterated logarithms, the rate of convergence
for the supremum distance is
\(n^{-1/2} \left( \log \log n \right)^{1/2}\) uniformly on \(h\) (see
Mason and Swanepoel (2013)).

The functions \(\beta^{(k)} \in \mathcal{L}\) and thus they also form a
Glivenko-Cantelli and Donsker class of functions for all
\(k = 1,2,\dots\). We have the following result.

\begin{proposition}[]\protect\hypertarget{prp-prp:betakest}{}\label{prp-prp:betakest}

Let Lemma~\ref{lem-lem:cdfcont}, and the conditions and results of
Proposition~\ref{prp-prp:cdfest} hold. Then \begin{align*}
\sup_{0 <h\leq b_n}\sup_{\alpha \in [0,1]} \vert \hat{\beta}^{(k)}_{h} (\alpha)- \beta^{(k)} (\alpha) \vert  = O \left( n^{-1/2} \left( \log \log n \right)^{1/2}\right),
\end{align*} for all \(k = 1,2,\dots\).

\end{proposition}

Having established the properties of the bias term in the previous
section and using the result in Proposition~\ref{prp-prp:betakest}, we
consider the properties of the term
\(\hat{\Lambda}_h^{\ast,N}(\alpha) - \Lambda^{\ast,N}(\alpha)\), which
we loosely refer to as the \emph{variance} term. We have the following
result.

\begin{lemma}[]\protect\hypertarget{lem-variance}{}\label{lem-variance}

Let Assumptions \ref{ass:attractive} and \ref{ass:priceregholder} hold.
Also, let \(r^\prime\) being uniformly bounded by a constant \(R\). Then
\[
\sup_{0 <h\leq b_n} \sup_{\alpha \in [0,1]}\vert \hat{\Lambda}_h^{\ast,N}(\alpha)  - \Lambda^{\ast,N}(\alpha)\vert = O \left( \frac{(N+1) \left( \log \log n \right)^{1/2}}{\sqrt{n}}\right).
\]

\end{lemma}

The variance term thus diverges with \(N^2\) for any given \(n\).
However, the bias term has an exponential decay as \(N\) grows. In
practice, one could balance these terms so that the estimator would have
an asymptotic bias. Another approach is to choose \(N\) to diverge more
slowly with \(n\) in a way that the asymptotic distribution is dominated
by the variance and we should not be too concerned with the bias when
conducting inference.

\subsection{Asymptotic distribution}\label{asymptotic-distribution}

Directly from the results in the previous section, we have the following
(see Thas (2009), Theorem 7.6, and Beare and Moon (2015)).

\begin{lemma}[]\protect\hypertarget{lem-lem:betaconv}{}\label{lem-lem:betaconv}

Let Assumptions \ref{ass:attractive} and \ref{ass:priceregholder} hold.
Then, \[
\sup_{0 <h\leq b_n} \sqrt{n} \left( \hat{\beta}_h^{(k)} - \beta^{(k)} \right) \Rightarrow \mathcal{B}_k,
\] where \begin{align*}
\mathcal{B}_k(\alpha) =& \mathbb{B}_1 \left( \beta^{(k)}(\alpha) \right)  + \sum_{j = 1}^{k-1} \mathbb{B}_1 \left( \beta^{(j)}(\alpha) \right) \frac{g_2\left( G_1^{-1} (\beta^{(k-j)}(\alpha ))\right)}{g_1\left( G_1^{-1} (\beta^{(k-j)}(\alpha ))\right)}  \\
& \qquad - \sum_{j = 0}^{k-1} \mathbb{B}_2 \left(\beta^{(j)} (\alpha)\right) \prod_{l = j}^{k-1}\frac{g_2 \left(G_1^{-1} (\beta^{(l)}(\alpha)) \right)}{g_1 \left( G_1^{-1} (\beta^{(l)} (\alpha)) \right)},
\end{align*}

where \(\mathbb{B}_1\) and \(\mathbb{B}_2\) are independent Brownian
bridges on \([0,1]\), for all \(k = 1,2,\dots\), and
\(\alpha \in [0,1]\).

\end{lemma}

From the results of Proposition~\ref{prp-prp:cdfest} and
Lemma~\ref{lem-lem:betaconv}, we have the following result about the
weak convergence of \(\hat{\Lambda}_h^{\ast,N}\) to
\(\Lambda^{\ast,N}\).

\begin{theorem}[]\protect\hypertarget{thm-thm:lambdaconv}{}\label{thm-thm:lambdaconv}

Let Assumptions \ref{ass:attractive}, \ref{ass:priceregholder}, and
Proposition~\ref{prp-prp:cdfest} and Lemma~\ref{lem-lem:betaconv} hold.
Then, \[
\sup_{0 <h\leq b_n} \sqrt{n} \left( \hat{\Lambda}_h^{\ast,N} - \Lambda^{\ast,N} \right) \Rightarrow \sum_{k = 1}^N \mathcal{B}_k r^\prime(\beta^{(k - 1)}), 
\] where \(\mathcal{B}_k\) is defined in Lemma~\ref{lem-lem:betaconv}
above.

\end{theorem}

Theorem~\ref{thm-thm:lambdaconv} follows directly from
Lemma~\ref{lem-lem:betaconv} and the proof is straightforward.

Finally the estimator of \(u\) is given by
\begin{equation}\phantomsection\label{eq-eq:uhat}{
\hat{u}^{N}_h(q) = P_2(q) - \int_{0}^q \hat{\Lambda}^{\ast,N}_h (\hat{G}_{h,2}(t)) dt,
}\end{equation} so that it inherits the properties of
\(\hat{\Lambda}^{\ast,N}\), except that the weak convergence should be
expressed with respect to the Brownian Bridge \(\mathbb{B}_{G_2}\) with
covariance function equal to \[
E \left[ \mathbb{B}_{G_2}(t) \mathbb{B}_{G_2}(s) \right]  = G_2 \left( t \wedge s \right) - G_2(t) G_2(s).
\]

The result is formally stated in the following theorem.

\begin{theorem}[]\protect\hypertarget{thm-thm:uconv}{}\label{thm-thm:uconv}

Let Assumptions \ref{ass:attractive}, \ref{ass:priceregholder}, and
Proposition~\ref{prp-prp:cdfest} and Lemma~\ref{lem-lem:betaconv} hold,
and \begin{align*}
\mathcal{B}_{G_2,k}(q)=& \mathcal{B}_{k}\left( G_2 (q) \right)\\
u^{N}(q) =& P_2(q) - \int_{0}^q \Lambda^{\ast,N} (G_2(t)) dt.
\end{align*} Then, \[
\sup_{0 <h\leq b_n} \sqrt{n} \left( \hat{u}^{N}_h(q) - u^{N}(q) \right) \Rightarrow -\sum_{k = 1}^N \mathbb{B}_{G_2,k} r^\prime(\beta^{(k - 1)} ). 
\]

\end{theorem}

The results of Theorem~\ref{thm-thm:lambdaconv} and
Theorem~\ref{thm-thm:uconv} apply for a fixed regularization parameter
\(N\). However, to get a consistent estimator of \(\Lambda\) and \(u\),
we must allow \(N\) to diverge with the sample size at an appropriate
rate. The following result combines the bias and variance bounds to give
a convergence rate for the regularized estimator.

\begin{lemma}[]\protect\hypertarget{lem-lem:consistency}{}\label{lem-lem:consistency}

Let the assumptions of Theorem~\ref{thm-thm:lambdaconv} hold, and let
\(N = N(n) \rightarrow \infty\) as \(n \rightarrow \infty\) with
\(N \asymp c \log(n)\) for some \(c > 0\) satisfying
\(\sqrt{n}\, \theta^{(N+1)\kappa} \rightarrow 0\). Then, \[
\sup_{\alpha \in (0,1)} \left\vert \hat{\Lambda}^{\ast,N}_h(\alpha) - \Lambda(\alpha) \right\vert = O_p \left( N \, n^{-1/2} \left( \log \log n \right)^{1/2} \right) = O_p \left( \log(n) \, n^{-1/2} \left( \log \log n \right)^{1/2} \right),
\] and \[
\sqrt{n} \left( \hat{\Lambda}^{\ast,N} (\alpha) - \Lambda (\alpha)\right) = \sqrt{n} \left( \hat{\Lambda}^{\ast,N} (\alpha) - \Lambda^{\ast,N}(\alpha) \right) ( 1 + o(1)).
\]

\end{lemma}

\begin{proof}
By the triangle inequality, \[
\left\vert \hat{\Lambda}^{\ast,N}_h(\alpha) - \Lambda(\alpha) \right\vert \leq \underbrace{\left\vert \hat{\Lambda}^{\ast,N}_h(\alpha) - \Lambda^{\ast,N}(\alpha) \right\vert}_{\text{variance}} + \underbrace{\left\vert \Lambda^{\ast,N}(\alpha) - \Lambda(\alpha) \right\vert}_{\text{bias}}.
\] By Lemma~\ref{lem-lembias}, the bias term is
\(O(\theta^{(N+1)\kappa})\), which decays exponentially in \(N\). By
Lemma~\ref{lem-variance}, the variance term is
\(O_p((N+1)\, n^{-1/2} (\log \log n)^{1/2})\). With
\(N \asymp c \log(n)\), the bias becomes
\(O(\theta^{c\kappa \log(n)}) = O(n^{c\kappa \log \theta})\), which is
\(o(n^{-1/2})\) provided \(c > 1/(2\kappa |\log \theta|)\). The variance
term becomes \(O_p(\log(n)\, n^{-1/2} (\log \log n)^{1/2})\). The
variance dominates, yielding the stated rate. Since
\(\sqrt{n}\, \theta^{(N+1)\kappa} \rightarrow 0\) by assumption, the
bias is asymptotically negligible relative to the variance, giving the
second display.
\end{proof}

The estimator thus achieves a near-parametric rate, with the \(\log(n)\)
factor reflecting the cost of regularization. The choice of \(c\)
controls the bias-variance tradeoff: larger \(c\) reduces the bias
faster but inflates the variance. A formal data-driven procedure for
selecting \(N\), remains an open question.

\section{Bootstrap}\label{sec-bootstrap}

Theorem~\ref{thm-thm:lambdaconv} and Theorem~\ref{thm-thm:uconv} are
difficult to use directly for inferential procedures. It is therefore
advisable to use the bootstrap to construct asymptotically valid
confidence intervals.

As we are working with empirical processes, we follow the approach in
Kosorok (2008) and use a version of the multiplier bootstrap to perform
asymptotically valid inference. Let
\(\xi = \lbrace \xi_1, \xi_2,\dots \rbrace\) be an infinite sequence of
non-negative iid random variables with mean \(1\) and variance \(1\) and
such that \[
\int_0^\infty \sqrt{P\left( \vert \xi \vert > x \right)} dx < \infty. 
\] (where the latter condition is satisfied whenever the
\(2 + \epsilon\) moment of \(\xi\) exists, for any \(\epsilon > 0\), see
Kosorok (2008), p.~20). We can, for instance, draw the random variable
\(\xi\) from a standard exponential distribution. Then we can replace
our estimators of \(G_1\) and \(G_2\) (and their corresponding
quantiles), with the following bootstrap estimate \begin{align*}
\hat{G}^\ast_{h,j}(q) =& \frac{1}{n_j} \sum_{i = 1}^{n_j} \frac{\xi_i}{\bar{\xi}} \bar{K} \left( \frac{Q_{ji} - q}{h_{j}}\right), \\
\hat{G}^{-1,\ast}_{h,j}(\alpha) =& \left( \hat{G}^\ast_{h,j}\right)^{-1}(\alpha),
\end{align*} for \(j = 1,2\), where \(\bar{\xi}\) is the sample mean of
\(\lbrace \xi_1,\dots,\xi_n \rbrace\). Under the same conditions as for
Theorem~\ref{thm-thm:lambdaconv} and Theorem~\ref{thm-thm:uconv}, this
version of the multiplier bootstrap converges (see Kosorok (2008),
Theorems 2.6 and 2.7, p.~20).

\section{Additional exogenous variables}\label{sec-addexo}

The utility function and the distribution of unobserved heterogeneity
can also depend on additional observable characteristics of the agent,
\(\lbrace X_j \in \mathbb{R}^p, j = 1,2 \rbrace\) that we wish to
include in the analysis.

We impose the following additional assumption.

\begin{assumption} \label{ass:xdist}
For $j = \lbrace 1,2\rbrace$, $X_j \in \mathcal{X} \subset \mathbb{R}^p$ and, for each $A \subset \mathcal{X}$, measurable, $P(X_1 \in A ) = 0$ if and only if $P(X_2 \in A ) = 0$.
\end{assumption}

Assumption \ref{ass:xdist} imposes that the \(X\) has identical support
across the two samples (have the same elements of null measure), so that
we can condition on the same values of \(X\) in both samples. One could
potentially relax this Assumption by letting
\(\mathcal{X} = \mathcal{X}_1 \cap \mathcal{X}_2\) to be non-empty.
Nonetheless, the identification argument below applies verbatim
conditional on all \(x \in \mathcal{X}\).

For \(j = \lbrace 1,2\rbrace\), the individual maximization problem
becomes \[
\max_Q u(Q,X) + Q\varepsilon - P(Q),
\] and the first order condition writes
\begin{equation}\phantomsection\label{eq-focxvar}{
u^\prime_Q(Q,X) - P^\prime(Q) = - \varepsilon,
}\end{equation} where \(u^\prime_Q\) denotes the first partial
derivative of \(u\) wrt \(Q\). The conditions for existence and
uniqueness of the solution to Equation~\ref{eq-focxvar} are the same as
above. However, the optimal quantity is not only determined by
\(\varepsilon\) and the utility function, but also depends on \(X\).

In the case of two observations, we have \[
\begin{aligned}
G_{1 \vert X} (q \vert X = x) =& P\left(\varepsilon \leq -\varphi_1(q,x) \vert X = x\right)  \\
G_{2 \vert X} (q \vert X = x) =& P\left(\varepsilon \leq -\varphi_2(q,x) \vert X = x \right), 
\end{aligned}
\] which, because of Assumption \ref{ass:xdist}, implies
\begin{equation}\phantomsection\label{eq-eq:condesteqnomeaserrs}{
\Lambda_{\varepsilon \vert X} \left( G_{2 \vert X} (q \vert X = x) \right) -  \Lambda_{\varepsilon  \vert X} \left( G_{1 \vert X} (q \vert X = x) \right)  = P^\prime_{2}(q) - P^\prime_{1}(q).
}\end{equation}

Assuming that \(G_{j\vert X}\) is strictly increasing in \(q\), for
\(j = \lbrace 1,2 \rbrace\),
\(G_{2 \vert X} (q \vert X = x ) > G_{1 \vert X} (q \vert X = x )\), for
all \(x \in \mathcal{X}\), and letting
\(q = G^{-1}_{2 \vert X}(\alpha \vert X = x)\),
\(\beta(\alpha,x) = G_{1 \vert X} \left( G^{-1}_{2 \vert X}(\alpha \vert X = x) \vert X = x \right)\)
and \[
r(\alpha,x) = P^\prime_{2}(G^{-1}_{2 \vert X}(\alpha \vert X = x)) - P^\prime_{1}(G^{-1}_{2 \vert X}(\alpha \vert X = x)),
\] we can write \[
\Lambda_{\varepsilon \vert X} \left(\alpha,x\right) = -\sum_{k = 0}^\infty r(\beta^{(k)}(\alpha,x),x).
\]

When \(X\) is a discrete random variable, we can stratify the sample wrt
the distribution of \(X\), and estimate the utility function and the
quantile function of \(\varepsilon\) conditional on \(X\). The
asymptotic theory we studied in the previous section can then be applied
after conditioning on \(X\).

When \(X\) is continuous, however, we need to smooth wrt its
distribution, which does not modify the estimation technique, but
requires a different approach to the asymptotic properties. We defer the
study of this case to further research.

\begin{remark}
In this and the following sections, we maintain the restriction that $\tau_1 = \tau_2 = 1$, independent of $X$. For identification of $\tau$, one could follow the approach outlined in Section \ref{sec-model} above, when $\tau_1$ is assumed to be invariant across $X$. If one wishes $\tau$ to vary as a function of $X$, identification depends on the properties of the conditional distribution of the quantities demanded in the two samples. If the intersection points between the cdfs change with $X$, then separate identification is possible.
\end{remark}

\section{Estimation of inverse demand function}\label{sec-priceest}

For each period of observation, our model is specified by the triplet
\((X,Q,P)\) and by a functional relationship between \(P\) and \(Q\),
the inverse demand function. We initially assume that the latter is
known to both the consumer and the econometrician. However, in many
situations, the inverse demand function is unknown to the
econometrician.

We can then model price and quantity as jointly determined conditional
on \(X\). The observed value of \((Q,P)\) can result, for instance, from
a negotiation between producer and consumer. We consider that the price
fixed by the producer does not depend on \(X\), but it is solely
determined by the quantity demanded. In particular, we make the
following assumption.

\begin{assumption} \label{ass:pricespec}
For $j = \lbrace 1,2\rbrace$, let $\Psi_{\theta_j}(q)$ be a known function of $q$ with unknown parameter $\theta_j$, and 
\begin{equation} \label{eq:pricespec}
P_j = \Psi_{\theta_j}(Q_j) + Q_j \eta_j,
\end{equation}
where $\eta_j$ is an IID disturbance whose distribution does not depend on $j$, and its realization is known to the decision maker and unknown to the econometrician. Its conditional distribution, $F_{\eta \Vert X} (\eta\vert X = x)$, is strictly increasing and such that $E\left[ \eta \vert X = x \right] = 0$, for all $x \in \mathcal{X}$.
\end{assumption}

Assumption \ref{ass:pricespec} stipulates that the observed price is
determined by Equation \ref{eq:pricespec}, and it is a function of the
quantity consumed, and of an unobservable component \(\eta\), which is
mean independent of \(X\). This random component can have several
interpretations. It could be related to unobservable taste of the
decision maker but also to their bargaining power. For instance, agents
who purchase higher quantities may be able to negotiate a lower price
per unit. \(\eta\) is centered wlog. If \(\eta\) is not centered, the
slope of the inverse demand curve can be adjusted by \(E[\eta]\).

The equilibrium quantity is then determined by the following
maximization problem \[
\max_Q u(Q,X) + Q\varepsilon - \Psi_{\theta}(Q) - Q \eta,
\] whose first order condition writes
\begin{equation}\phantomsection\label{eq-focmeserr}{
u^\prime_Q(Q,X) - \Psi^\prime_{\theta}(Q) = - \varepsilon + \eta.
}\end{equation} This condition implies that the equilibrium quantity
depends on \(X\) both through the utility function and through the
conditional distribution of \(\varepsilon\) given \(X\).

If the parameters \(\theta\) were known, then \(\eta\) can also be taken
as known. In this case, thanks to Assumption \ref{ass:xdist}, one could
perform estimation conditional on \(X\), as explained in the previous
section.

In the case of two observations, we have \[
\begin{aligned}
G_{1 \vert X} (q \vert X = x) =& \Pr\left(\varepsilon - \eta \leq -\varphi_1(q,x) \vert X = x\right)  \\
G_{2 \vert X} (q \vert X = x) =& \Pr\left(\varepsilon - \eta \leq -\varphi_2(q,x) \vert X = x \right), 
\end{aligned}
\] which implies \begin{equation}\phantomsection\label{eq-eq:condesteq}{
\Lambda_{\zeta \vert X} \left( G_{2 \vert X} (q \vert X = x) \right) -  \Lambda_{\zeta \vert X} \left( G_{1 \vert X} (q \vert X = x) \right)  = P^\prime_{2}(q) - P^\prime_{1}(q),
}\end{equation} with \(\zeta = \varepsilon - \eta\), and
\(\Lambda_{\zeta \vert X}\), the conditional quantile function of the
composite error \(\zeta\) conditional on \(X\). Similar identification
conditions as the one outlined above would hold conditionally on \(X\).
The quantile function of \(\zeta\) conditional on \(X\) and the utility
function are estimated exactly as before.

As discussed in Section~\ref{sec-model}, the object we identify is the
distribution of \(\zeta\) conditional on \(X\) and \(Q > 0\). This is
because we are only able to observe those consumers who choose a
positive quantity. While we do not explicitly condition on \(Q >0\)
below, this is an important detail that we would like the reader to keep
in mind. Our identification of the distribution of \(\varepsilon\) given
\(X\) and \(Q > 0\) is based on the estimation of the error term
\(\zeta\), as follows.

From Equation~\ref{eq-focmeserr}, and for known \(\theta\) and \(\eta\),
the structural residuals can be obtained as
\begin{equation}\phantomsection\label{eq-geterrterm}{
\hat{\zeta}_i(\theta) = -\left( \hat{u}^\prime_Q(Q_{i},X_i) - \Psi^\prime_{\theta}(Q_i) \right),
}\end{equation}

and thus we have
\(\hat{\varepsilon}_i(\theta) = \hat{\zeta}_i(\theta) + \eta_i\), and
the conditional cdf is obtained using a simple nonparametric approach
(Li and Racine (2008)).

However, \(\theta\) and \(f_{\eta\vert X}\) are unknown and need to be
estimated. A major hurdle is that the quantity chosen by the decision
maker (and observed by the econometrician) is a function of the
composite shock \(\zeta\). In Equation \ref{eq:pricespec}, quantity is
therefore an endogenous variable.

The endogeneity issue can be solved by using \(X\) as an instrumental
variable (IV) in Equation \ref{eq:pricespec}. In particular, we have
that \[
E\left[ \frac{P_j}{Q_j} \vert X_j = x \right] =  E \left[  \frac{\Psi_{k,j}(Q_j)}{Q_j} \vert X_j = x \right] \theta_j + E\left[ \eta_j \vert X = x \right] =  E \left[  \frac{\Psi_{\theta_j}(Q_j)}{Q_j} \vert X_j = x \right],
\] where the last equality follows from Assumption \ref{ass:pricespec},
which specifies an exogeneity condition for \(X\). Moreover, as \(X\)
enters the utility function, it also determines the chosen quantity, and
it is therefore a relevant IV. Estimation of the inverse demand function
using a parametric IV model allows us to obtain an estimator of the
density of \(\eta\) conditional on \(X\). Depending on the specification
of \(\Psi_{\theta}(q)\), we could use several approaches for estimation.
If the function \(\Psi_{\theta_j}(q)\) is linear in parameters, we can
directly use an instrumental variable approach (see
Example~\ref{exm-exinsvar1} below). If, instead, \(\Psi_{\theta_j}(q)\)
is nonlinear in parameters, one can use a nonlinear IV model that can be
estimated by the method of moments (see Example~\ref{exm-exinsvar2}
below).

One can then construct the residuals \[
\hat\eta_j = \frac{P_j}{Q_j} - \frac{\Psi_{\hat\theta_j}(Q_j)}{Q_j}, \text{ for } j = 1,2.
\]

The estimator of the difference in the derivative of the price functions
at a point \(q = G_1^{-1} (\alpha)\) is then given by
\begin{equation}\phantomsection\label{eq-diffpriceest}{
\hat{r}(\alpha) = \Psi^\prime_{\hat{\theta}_1}(G_1^{-1} (\alpha)) - \Psi^\prime_{\hat{\theta}_2}(G_1^{-1} (\alpha)).
}\end{equation}

Finally, \begin{equation}\phantomsection\label{eq-fepscond_finalest}{
\hat{F}_{\varepsilon \vert X} \left( \bar{\varepsilon} \vert X = x \right) = \frac{\sum_{i = 1}^{n_j} \mathbbm{1} \left(\hat{\zeta}_{ji} + \hat{\eta}_{ji} \leq \bar{\varepsilon} \right) \mathbbm{1} \left(X_{ji} = x\right)}{\sum_{i = 1}^{n_j} \mathbbm{1} \left(X_{ji} = x\right)},
}\end{equation} where \(j\) could be either \(1\) or \(2\) (depending on
the sample we use to estimate the utility function).

\begin{example}[]\protect\hypertarget{exm-exinsvar1}{}\label{exm-exinsvar1}

Let \(\Psi_{\theta_j}(q) = (q,\dots,q^d) \theta_j\), a \(d\)-th order
polynomial in \(q\), and \(X\) be a discrete variable taking
\(\ell \geq 2\) distinct values \(\lbrace 0,\dots,\ell-1 \rbrace\), and
\(\Xi_{j}(x)  = \left( \mathbbm{1}(X = 0), \dots, \mathbbm{1}(X = \ell -1) \right)^\prime\),
where \(\ell \geq d\), and
\(\boldsymbol\Xi_j = E\left[ \Xi^\top_{j}(X) \Xi_{j}(X) \right]\), a
symmetric, positive definite matrix. Then, the estimand of \(\theta_j\)
is directly written as
\begin{equation}\phantomsection\label{eq-thetajestmd}{
\begin{aligned}
\theta_j =& \left( E\left[ \left(\frac{\Psi_{\theta_j}(Q_j)}{Q_j} \right)^\top \Xi_{\ell,j}(X_j) \right] \boldsymbol\Xi^{-1}_j E\left[  \Xi^\top_{\ell,j}(X_j) \frac{\Psi_{\theta_j}(Q_j)}{Q_j} \right] \right)^{-1} \\
& \qquad \left( E\left[ \left(\frac{\Psi_{\theta_j}(Q_j)}{Q_j} \right)^\top \Xi_{\ell,j}(X_j) \right] \boldsymbol\Xi^{-1}_j E\left[  \Xi^\top_{\ell,j}(X_j) \frac{P_j}{Q_j} \right] \right),
\end{aligned}
}\end{equation} for \(j = \lbrace 1,2 \rbrace\). Under the maintained
parametric assumptions, the proposed IV estimator is consistent and
asymptotically normal. \(\tau_1\) is identified by the ratio of the
intercepts in the two samples, when there is a stochastic ordering
between the conditional cdfs for all values of \(x\).

\end{example}

\begin{example}[]\protect\hypertarget{exm-exinsvar2}{}\label{exm-exinsvar2}

Let us take \(\Psi_{\theta_j}(q) = \theta_{0,j} + exp(-\theta_{1,j}q)\).
We can then write \[
(\hat{\theta}_{0,j},\hat{\theta}_{1,j}) = \argmin_{(\theta_0,\theta_1)} \left( \sum_{i = 1}^{n_j} (P_{ji} - \theta_{0,j} - exp(-\theta_{1} Q_{ji})) X^\prime_{ji}\right) \mathcal{W} \left( \sum_{i = 1}^{n_j} (P_{ji} - \theta_{0,j} - exp(-\theta_{1} Q_{ji})) X^\prime_{ji}\right)^\prime,
\] for \(j = 1,2\), where \(\mathcal{W}\) is an appropriately chosen
weighting matrix.

\end{example}

\section{Monte-Carlo Simulations}\label{sec-montecarlo}

We consider the following data-generating process.
\(\varepsilon \sim F_\varepsilon\), with density equal to
\(f_\varepsilon(\varepsilon) = 1.5 \sqrt{\varepsilon}\) supported on
\([0,1]\). The utility function is given by \(u(Q) = 2Q - Q^2\). We have
two simulation designs, depending on the specification of the price
function.

\begin{enumerate}
\def\labelenumi{\arabic{enumi}.}
\tightlist
\item
  \textbf{Design 1}
\end{enumerate}

\begin{align*}
P_1(Q) =& Q - \frac{Q^2}{2} + \frac{Q^3}{6}\\
P_2(Q) =& 2Q - \frac{Q^2}{2},
\end{align*}

with \(\tau_1 = 2\), and \(\tau_2 = 1\), in a way that the equilibrium
quantities \(Q^\ast_1 = \sqrt{\varepsilon}\) and
\(Q^\ast_2 = \varepsilon\), respectively. In this case, \(G_1\)
stochastically dominates \(G_2\) over \(\mathcal{Q}\).

\begin{enumerate}
\def\labelenumi{\arabic{enumi}.}
\setcounter{enumi}{1}
\tightlist
\item
  \textbf{Design 2}
\end{enumerate}

\begin{align*}
P'_1(Q) =& \frac{1}{2}\left(1 + \sum_{k=0}^{4}\binom{9}{k}Q^k(1-Q)^{9-k}\right), \quad P_1(Q) = \int_0^Q P'_1(q)\,dq\\
P_2(Q) =& 2Q - \frac{Q^2}{2},
\end{align*}

with \(\tau_1 = 2\), and \(\tau_2 = 1\). The equilibrium quantity
\(Q^\ast_1\) is obtained by solving
\(u'(Q^\ast_1) = \tau_1 P'_1(Q^\ast_1) - \varepsilon\) numerically,
while \(Q^\ast_2 = \varepsilon\) as in Design 1. In this case, \(G_1\)
and \(G_2\) intersect at exactly one point in \(\mathcal{Q}\).

The number of observations per sample is set to
\(n = n_1 = n_2 = \lbrace 500, 1000, 2500 \rbrace\). We consider the
case in which the price functions are known, and we wish to estimate the
utility function and the quantile function of \(\varepsilon\) from data.
We also include in the Appendix a set of simulations for the case when
the price function is estimated following the approach in
Section~\ref{sec-priceest}.

To generate the data on quantity, we simulate \(2n\) observations from
the distribution of \(\varepsilon\), and then we fix the quantities
\(Q^\ast_1\) and \(Q^\ast_2\) according to their equilibrium values.
\(G_1\) and \(G_2\) are estimated using a Gaussian kernel, with
bandwidths chosen by least-squares cross-validation. The estimated
quantile functions
\(\hat{G}_{h,j}^{-1}(\alpha) = \min_q \lbrace \hat{G}_{h,j}(q) \geq \alpha \rbrace\).

For \textbf{Design 1}, our starting value is given by
\(r(\alpha) = \tau_1 P^\prime_1(\hat{Q}_\alpha) - P^\prime_2(\hat{Q}_\alpha)\),
for a grid of percentiles \(\hat{Q}_\alpha\) taken from the empirical
distribution of \(Q^\ast_1\). At step 1, we estimate \(\beta(\alpha)\)
as \[
\hat{\beta}^{(1)}_h (\alpha) = \left( \hat{G}_{h,2} \right)^{-1} \left( \hat{G}_{h,1} (\hat{Q}_\alpha) \right), 
\] and then compute \[
\hat{\Lambda}^{\ast,1}(\alpha) = -r(\alpha) - r\left( \hat{\beta}^{(1)}_h (\alpha) \right).
\] We then iterate the approach for \(k = 2,\dots,N\), in a way that \[
\hat{\beta}^{(k)}_h (\alpha) = \left( \hat{G}_{h,2} \right)^{-1} \left( \hat{G}_{h,1} (\hat{\beta}^{(k-1)}_h (\alpha) ) \right), 
\] and \[
\hat{\Lambda}^{\ast,k}(\alpha) = -r(\alpha) - \sum_{j = 1}^{k} r\left( \hat{\beta}^{(j)}_h (\alpha) \right).
\]

To choose the number of iterations \(N\), we proceed as follows. As
\(G_1\) stochastically dominates \(G_2\) over \(\mathcal{Q}\), we have
that \(0\) is an attractive point for \(\hat{\beta}^{(k)}_h (\alpha)\)
as \(k\) increases. Hence, at each step, we compute the variance of
\(\hat{\beta}^{(k)}_h (\alpha)\) and we stop when its variance is
smaller than a given tolerance level, which we take to be equal to
\(10^{-10}\). We also cap the number of iterations at \(5 \log(n)\) to
control the approximation error.

Finally, the utility function is estimated as \[
\hat{u}^N_h(Q) = P_2(Q) - \int_0^Q \hat{\Lambda}^{\ast,N}\left( \hat{G}_{h,2}(q) \right) dq.
\]

For \textbf{Design 2}, we proceed in a similar way, except that we now
have to split \(\mathcal{Q}\) into two subsets, one in which \(G_1\)
stochastically dominates \(G_2\) and another one in which the opposite
is true. We choose the splitting point based on the difference between
the price derivatives, and conduct separate estimation of \(\Lambda\)
and \(u\) over these two nonoverlapping subsets.

To provide a sense of the variability in our estimator, we provide the
result for 20 random samples of size \(n = 1000\) for Design 1 and 2 in
Figure~\ref{fig-figLambda0} and Figure~\ref{fig-figLambda1},
respectively.

\begin{figure}

\begin{minipage}{0.50\linewidth}

\centering{

\includegraphics{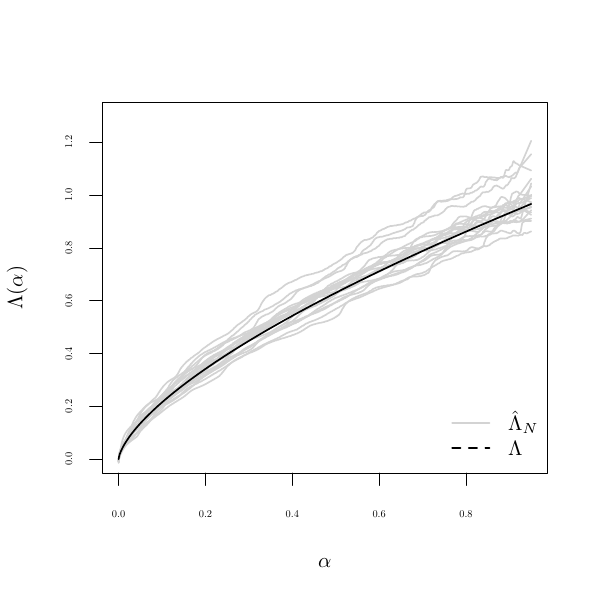}

}

\subcaption{\label{fig-figLambda0-1}Estimation of
\(\Lambda_{\varepsilon}\)}

\end{minipage}%
\begin{minipage}{0.50\linewidth}

\centering{

\includegraphics{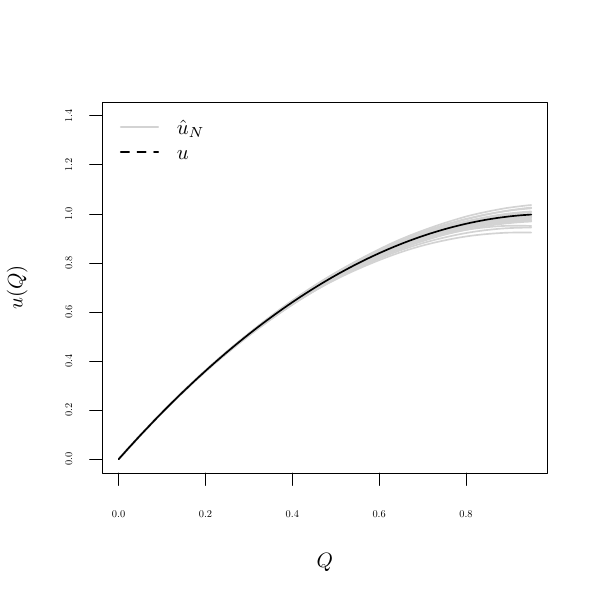}

}

\subcaption{\label{fig-figLambda0-2}Estimation of \(u\)}

\end{minipage}%

\caption{\label{fig-figLambda0}Estimation of the quantile function of
\(\varepsilon\) (a) and the utility function (b) for Design 1 for 20
randomly selected sample of size \(n = 1000\).}

\end{figure}%

\begin{figure}

\begin{minipage}{0.50\linewidth}

\centering{

\includegraphics{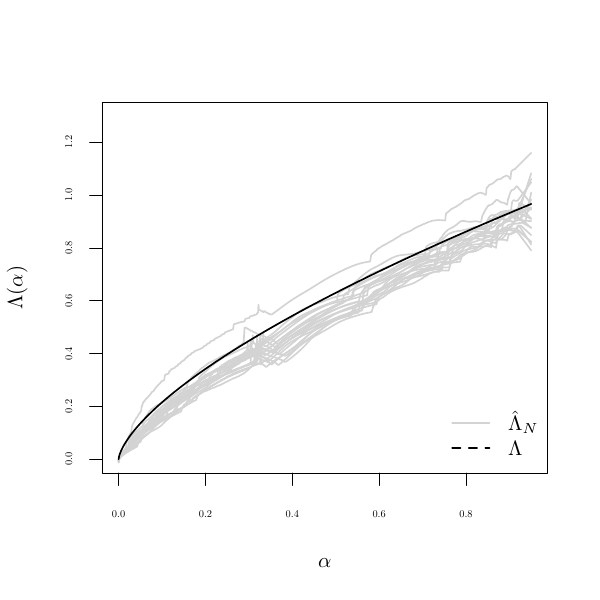}

}

\subcaption{\label{fig-figLambda1-1}Estimation of
\(\Lambda_{\varepsilon}\)}

\end{minipage}%
\begin{minipage}{0.50\linewidth}

\centering{

\includegraphics{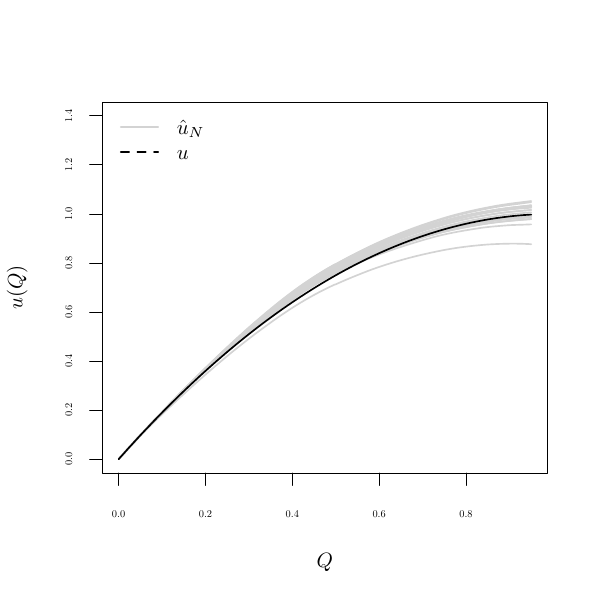}

}

\subcaption{\label{fig-figLambda1-2}Estimation of \(u\)}

\end{minipage}%

\caption{\label{fig-figLambda1}Estimation of the quantile function of
\(\varepsilon\) (a) and the utility function (b) for Design 2 for 20
randomly selected sample of size \(n = 1000\).}

\end{figure}%

Bootstrap confidence intervals are obtained using the strategy outlined
in Section~\ref{sec-bootstrap}. Results for one randomly selected sample
with \(n = 1000\) are given in Figure~\ref{fig-figLambda2} and
Figure~\ref{fig-figLambda3}. We also visually compare the Monte-Carlo
pointwise confidence intervals with their bootstrap counterpart to
assess their accuracy, and coverage rates. The bootstrap CIs are
slightly conservative in this setting.

\begin{figure}

\begin{minipage}{0.50\linewidth}

\centering{

\includegraphics{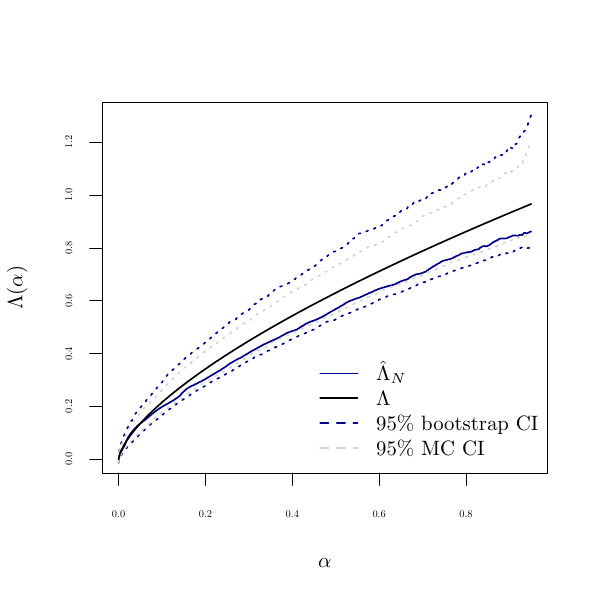}

}

\subcaption{\label{fig-figLambda2-1}Estimation of
\(\Lambda_{\varepsilon}\)}

\end{minipage}%
\begin{minipage}{0.50\linewidth}

\centering{

\includegraphics{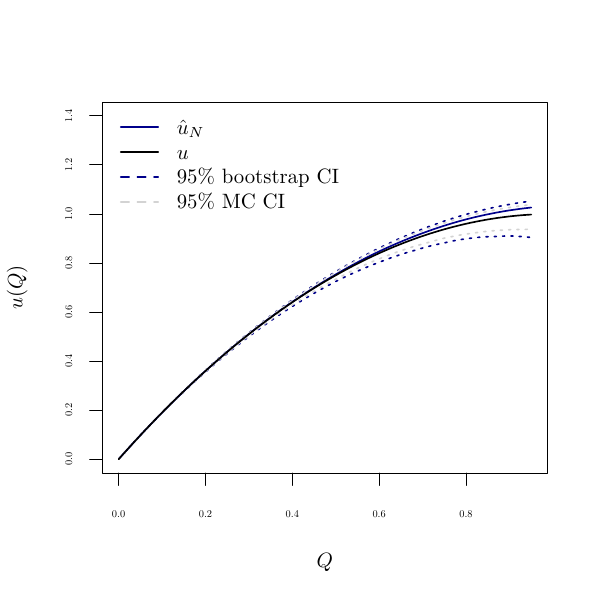}

}

\subcaption{\label{fig-figLambda2-2}Estimation of \(u\)}

\end{minipage}%

\caption{\label{fig-figLambda2}Estimation of the quantile function of
\(\varepsilon\) (a) and the utility function (b) for Design 1 for a
randomly selected sample of size \(n = 1000\).}

\end{figure}%

\begin{figure}

\begin{minipage}{0.50\linewidth}

\centering{

\includegraphics{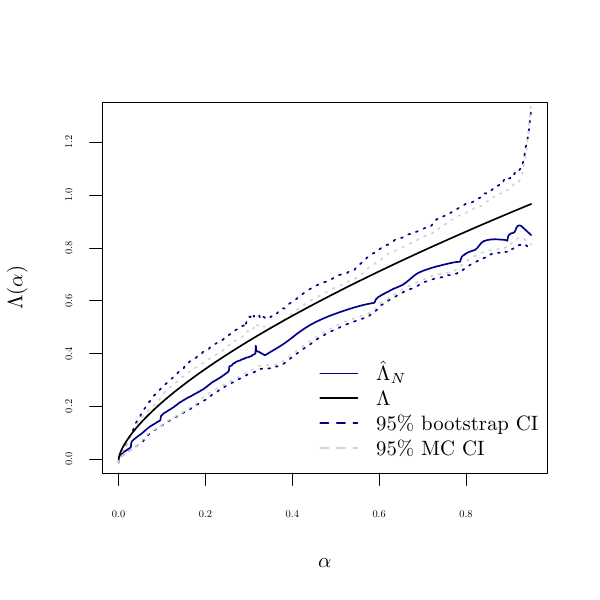}

}

\subcaption{\label{fig-figLambda3-1}Estimation of
\(\Lambda_{\varepsilon}\)}

\end{minipage}%
\begin{minipage}{0.50\linewidth}

\centering{

\includegraphics{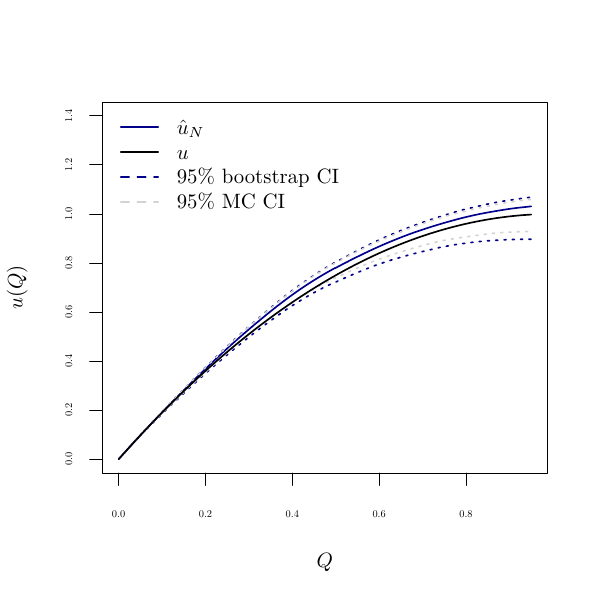}

}

\subcaption{\label{fig-figLambda3-2}Estimation of \(u\)}

\end{minipage}%

\caption{\label{fig-figLambda3}Estimation of the quantile function of
\(\varepsilon\) (a) and the utility function (b) for Design 2 for a
randomly selected sample of size \(n = 1000\).}

\end{figure}%

Table~\ref{tbl-simresults} gives the average maximum absolute error of
the estimator for both simulation designs and all sample sizes. As
expected the error of our estimator decreases as the sample size
increases.

\begin{table}

\caption{\label{tbl-simresults}Average maximum absolute deviation of the
estimator}

\centering{

\begin{tabular}{rrrrrrr}
\toprule
\multicolumn{1}{c}{} & \multicolumn{3}{c}{Design 1} & \multicolumn{3}{c}{Design 2} \\
\cmidrule(l{3pt}r{3pt}){2-4} \cmidrule(l{3pt}r{3pt}){5-7}
n & $\hat{\tau}_1$ & $\hat{\Lambda}_N$ & $\hat{u}_N$ & $\hat{\tau}_1$ & $\hat{\Lambda}_N$ & $\hat{u}_N$\\
\midrule
500 & 0 & 0.1210 & 0.0294 & 0 & 0.1420 & 0.0421\\
1000 & 0 & 0.0912 & 0.0206 & 0 & 0.1249 & 0.0361\\
2500 & 0 & 0.0614 & 0.0125 & 0 & 0.1346 & 0.0292\\
\bottomrule
\end{tabular}

}

\end{table}%

Finally, and only for Design 1, we compare our iterative estimator of
the quantile function of heterogenous types, with an estimator which is
based on the first derivative of \(\lambda\) obtained by
Tikhonov-regularization. That is \[
\hat{\Lambda}_{Tik}(\alpha) = \int_{0}^\alpha \hat{\lambda}_{Tik}(t) dt =  \int_{0}^\alpha \left[ \left( \rho I + \hat{K}^\ast \hat{K} \right)^{-1} \hat{K}^\ast r \right](t) dt,
\] where \(\hat{K}\) is a finite approximation to the integral operator,
\(\hat{K}^\ast\) its adjoint, and \(\rho\) a regularization parameter
that we choose proportional to the inverse of the number of iterations
for each simulated DGP.

\begin{table}

\caption{\label{tbl-simresults-tik}Comparison between our iterative
approach and a Tikhonov-regularized estimator of \(\Lambda\)}

\centering{

\begin{tabular}{rrrrr}
\toprule
\multicolumn{1}{c}{n} & \multicolumn{2}{c}{$\hat{\Lambda}_N$} & \multicolumn{2}{c}{$\hat{\Lambda}_{Tik}$} \\
\cmidrule(l{3pt}r{3pt}){1-1} \cmidrule(l{3pt}r{3pt}){2-3} \cmidrule(l{3pt}r{3pt}){4-5}
 & Bias & Std Dev & Bias & Std Dev\\
\midrule
500 & 0.0243 & 0.1210 & 0.0362 & 0.0801\\
1000 & 0.0140 & 0.0915 & 0.0331 & 0.0587\\
2500 & 0.0057 & 0.0614 & 0.0273 & 0.0381\\
\bottomrule
\end{tabular}

}

\end{table}%

The results reported in Table~\ref{tbl-simresults-tik} confirm that this
estimator also has good properties in finite sample. In the specific
case of our simulation study, the Tikhonov-regularized estimator has
smaller mean absolute error than our iterative estimator. We conjecture
that this difference is due to the fact that the Tikhonov estimator
takes additional advantage of the differentiability of the quantile
function, which is neither required nor used by the iterative estimator.

\section{Empirical Application}\label{sec-application}

We obtained data from a European mail carrier, which contain information
about unaddressed advertising (or admail). Unaddressed advertising is a
form of direct mail where marketing materials are distributed to a
specific geographic area without the need to be addressed to individual
residences or businesses. This method aims to reach a broad audience
within a particular area, increasing brand visibility and potentially
generating higher response rates compared to targeted mail.

Negotiations about the quantity and unit price of unaddressed mail are
conducted between the mail carrier and individual businesses. The
contract can include one or more geographical areas. We have information
about the nature of the business (a broadly defined sector of
operation), the total price paid and the quantity negotiated. If a
single contract involves multiple geographical areas, we have
information about price/quantity for each geographical area, and we
consider these as separate contracts for simplicity.

The data are a repeated cross-section that span 6 years of observations
(from 2008 to 2013). Observations from the first and last year are
incomplete, and we therefore use data from 2009 and 2012 for our
analysis. We refer to 2009 as \(t = 1\), and 2012 as \(t  = 2\)
henceforth.

We remove from the sample all the sectors that cannot be found in both
years (as we cannot apply our proposed framework in that case), and we
also remove sectors with fewer than 400 observations in total between
the two periods. Also, we remove sectors with outliers (contracts with
very small or extremely large quantities, based on their interquartile
range). We end up with 15341 observations at \(t = 1\), and 17875
observations at \(t = 2\).

Table~\ref{tbl-summstats} reports summary statistics for the four chosen
sectors: breakdown services, car dealers, electronics shops, and health
aid products. For each sector and year, we report the number of
contracts, the mean and standard deviation of quantity (in thousands of
items delivered) and of unit price (in pence per item).

\begin{table}

\caption{\label{tbl-summstats}Summary statistics for the four selected
sectors, 2009 (\(t=1\)) and 2012 (\(t=2\)).}

\centering{

\begin{tabular}{lrrrrrrrrrr}
\toprule
\multicolumn{1}{c}{ } & \multicolumn{5}{c}{2009 ($t = 1$)} & \multicolumn{5}{c}{2012 ($t = 2$)} \\
\cmidrule(l{3pt}r{3pt}){2-6} \cmidrule(l{3pt}r{3pt}){7-11}
  & $n$ & $\bar{Q}$ & $\sigma_Q$ & $\bar{P}$ & $\sigma_P$ & $n$ & $\bar{Q}$ & $\sigma_Q$ & $\bar{P}$ & $\sigma_P$\\
\midrule
Breakdown & 84 & 176.36 & 237.08 & 7.39 & 9.39 & 323 & 302.87 & 256.09 & 11.82 & 9.88\\
Car Dealers & 621 & 108.46 & 151.37 & 5.48 & 7.00 & 363 & 142.09 & 181.71 & 5.78 & 7.28\\
Elec Shops & 169 & 115.28 & 158.77 & 5.57 & 7.08 & 393 & 189.13 & 229.33 & 8.13 & 9.91\\
Health Aid & 453 & 193.88 & 214.47 & 8.90 & 9.66 & 454 & 284.27 & 290.76 & 10.43 & 10.60\\
\bottomrule
\end{tabular}

}

\end{table}%

The price schedule is unobserved, and we estimate it assuming that the
utility function and the consumer types can be sector-specific. That is,
we assume that each firm solves

\[
\max_{q \in \mathcal{Q}} u(Q,X) - \Psi_{\theta_j} (Q)  + Q \zeta
\]

where \(\zeta = \varepsilon - \eta\) is the composite unobservable
defined in Section~\ref{sec-priceest}, with \(\varepsilon\) capturing
unobserved preference heterogeneity and \(\eta\) the price disturbance.
We estimate the price function, recover the structural residuals, and
apply the iterative procedure described in Section~\ref{sec-estimation}
and Section~\ref{sec-priceest} to obtain estimates of the utility
function and the distribution of unobserved types for each sector.
Results are displayed in Figure~\ref{fig-figDataApp}.

\begin{figure}[ht]

\begin{minipage}{0.50\linewidth}

\centering{

\includegraphics{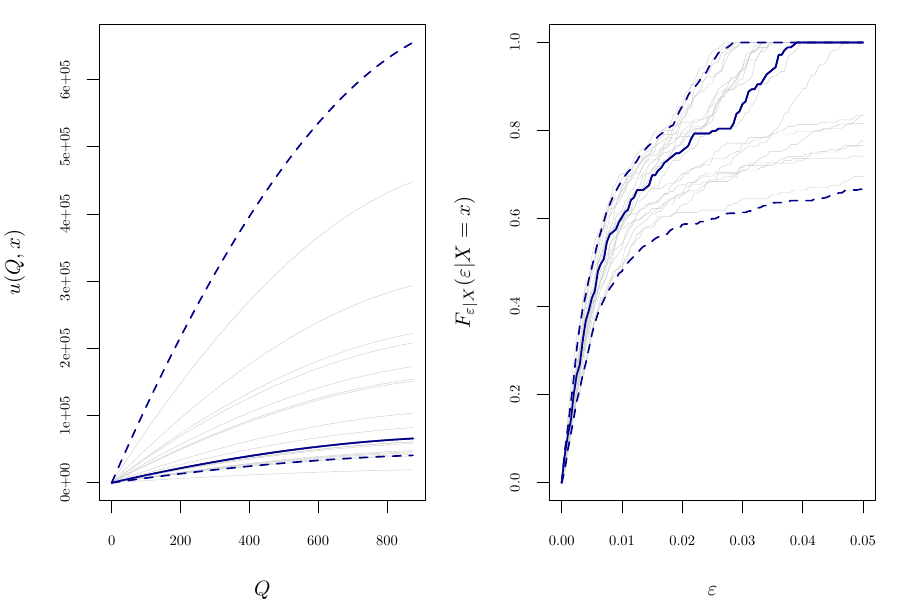}

}

\subcaption{\label{fig-figDataApp-1}Breakdown}

\end{minipage}%
\begin{minipage}{0.50\linewidth}

\centering{

\includegraphics{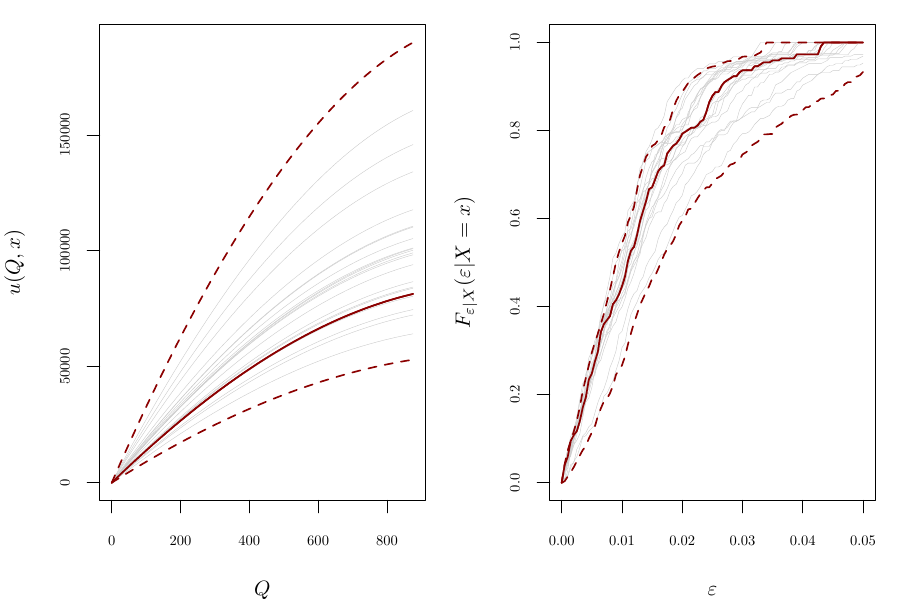}

}

\subcaption{\label{fig-figDataApp-2}Car Dealers}

\end{minipage}%
\newline
\begin{minipage}{0.50\linewidth}

\centering{

\includegraphics{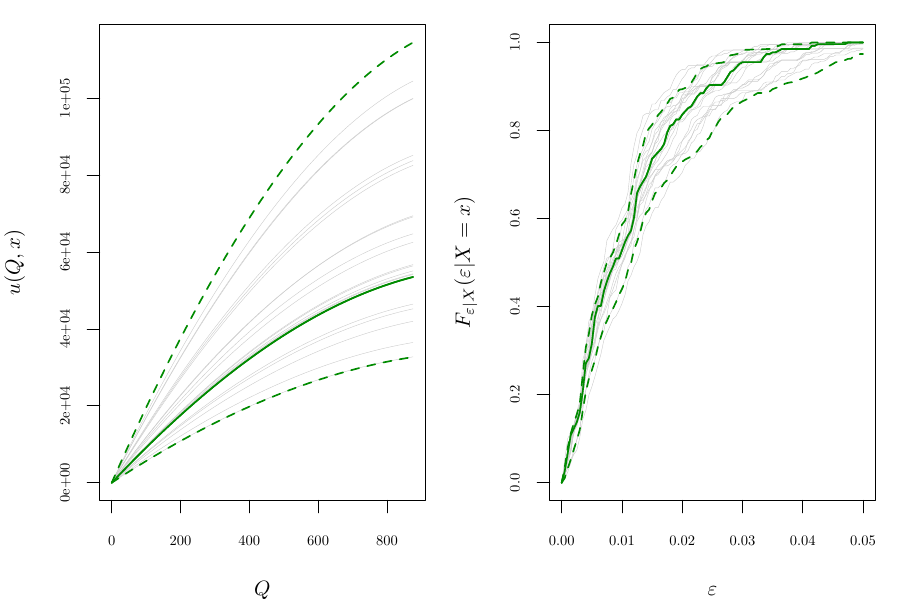}

}

\subcaption{\label{fig-figDataApp-3}Elec Shops}

\end{minipage}%
\begin{minipage}{0.50\linewidth}

\centering{

\includegraphics{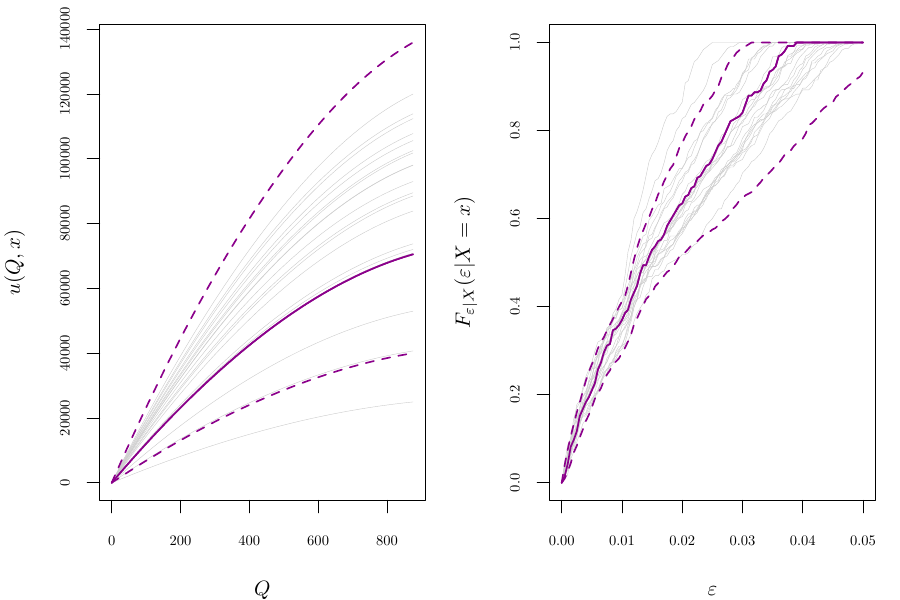}

}

\subcaption{\label{fig-figDataApp-4}Health Aid}

\end{minipage}%

\caption{\label{fig-figDataApp}Estimation of \(u\) (left) and
\(F_{\varepsilon \vert X}\) (right) for the four selected sectors. The
dashed lines are the \(90\%\) bootstrap confidence intervals. The grey
lines are 20 random bootstrap realizations.}

\end{figure}%

\subsection{Demand Elasticities}\label{demand-elasticities}

Once the utility function \(u(Q; X)\) and the price schedules \(P(Q)\)
have been estimated, one of the potential objects of interest are demand
elasticities. Under nonlinear pricing, the notion of a price elasticity
of demand is not uniquely defined: since the price schedule is a
function rather than a scalar, different counterfactual experiments
yield conceptually distinct elasticity measures (Wilson 1993; Tirole
1988). We consider two natural counterfactual experiments, based on
transformations of the price schedule, and derive the corresponding
elasticities using the implicit function theorem applied to the
first-order condition.

\subsubsection{Elasticity of demand to a price
perturbation}\label{elasticity-of-demand-to-a-price-perturbation}

The interior optimum is determined by the first-order condition
\begin{equation}\phantomsection\label{eq-foc-elast}{
u'(Q;\,x) - P'(Q) + \varepsilon = 0.
}\end{equation} We modify the price schedule \(P\) using a one-parameter
family \(\Pi_\sigma(P)\), indexed by \(\sigma \in \mathbb{R}\), such
that \(\Pi_{\sigma_0}(P) = P\) for \(\sigma_0\) given. The elasticity of
demand with respect to \(\sigma\) is defined as
\begin{equation}\phantomsection\label{eq-elast-def}{
e^{(\sigma)}(Q,\, x) \;=\; \frac{dQ}{d\sigma} \frac{\sigma}{Q}.
}\end{equation}

Under the modified schedule, the first-order condition becomes
\(u'(Q;\,x) - \Pi_\sigma\left(P(Q)\right)^\prime + \varepsilon = 0\),
where \(\Pi_\sigma\left(P(Q)\right)^\prime\) is the first derivative of
the composite function \(\Pi_\sigma(P(Q))\) wrt \(Q\). Total
differentiation of Equation~\ref{eq-foc-elast} with respect to
\(\sigma\) at fixed \(\varepsilon\) yields \[
\bigl[u''(Q;\,x) - \Pi_\sigma\left(P(Q)\right)^{\prime \prime}\bigr]\, dQ \;=\; \frac{\partial}{\partial \sigma}\, \Pi_\sigma\left(P(Q)\right)^\prime\, d\sigma,
\] so that \begin{equation}\phantomsection\label{eq-elast}{
e^{(\sigma)}(Q,\, x) \;=\; \frac{\dfrac{\partial}{\partial \sigma}\, \Pi_\sigma\left(P(Q)\right)^\prime}{\bigl[u''(Q;\,x) - \Pi_\sigma\left(P(Q)\right)^{\prime \prime}\bigr]}\frac{\sigma}{Q}.
}\end{equation}

At the point, \(\sigma = \sigma_0\),
\(\Pi_\sigma\left(P(Q)\right)^\prime = P^\prime (Q)\), and
\(\Pi_\sigma\left(P(Q)\right)^{\prime \prime} = P^{\prime \prime} (Q)\).
This yields \begin{equation}\phantomsection\label{eq-elast-2}{
e^{(\sigma_0)}(Q,\, x) \;=\; \frac{\dfrac{\partial}{\partial \sigma}\, \Pi_\sigma\left(P(Q)\right)^\prime \bigg|_{\sigma = \sigma_0}}{Q \bigl[u''(Q;\,x) - P^{\prime \prime}(Q)\bigr]}.
}\end{equation}

Under the conditions in Assumption \ref{ass:regcond},
\(u''(Q,x)-P''(Q)<0\). For any price-increasing perturbation, the
partial derivative
\(\frac{\partial}{\partial \sigma}\,\Pi_\sigma\left(P(Q)\right)^\prime\)
is positive, and the elasticity is therefore negative, confirming that
demand is downward-sloping in the price parameter.

The second derivative \(u''(Q,x)\) is recovered using the relationship
between the conditional density of observed quantities and the quantile
function of the unobserved heterogeneity. Specifically, differentiating
Equation~\ref{eq-utility-identification} twice with respect to \(q\)
yields \[
u''(Q,x) \;=\; P''(Q) \;-\; \frac{g_{Q \vert X}(Q \mid X=x)}{f_{\varepsilon \vert X}\!\bigl(\Lambda(G_{Q\vert X}(Q \mid X=x))\mid X=x\bigr)},
\] where \(g_{Q \vert X}(\cdot \mid X)\) is the conditional density of
\(Q\), \(f_{\varepsilon \vert X}(\cdot \mid X)\) is the conditional pdf
of \(\varepsilon\), and \(\Lambda(G_{Q\vert X}(Q \mid X=x))\) is the
quantile function of \(\varepsilon\) conditional on sector \(X\).

\subsubsection{Level and Curvature
Counterfactuals}\label{level-and-curvature-counterfactuals}

\textbf{Level change.} Consider the family
\(\Pi_\sigma(P)(Q) = \sigma \, P(Q)\), with reference value
\(\sigma_0 = 1\). This rescales the schedule uniformly, capturing a
change in the overall \emph{level} of prices. We have
\(\Pi_\sigma(P)'(Q) = \sigma\, P'(Q)\), so
\(\frac{\partial}{\partial \sigma}\, \Pi_\sigma(P(Q))^\prime = P'(Q)\).
Evaluating Equation~\ref{eq-elast} at \(\sigma_0 = 1\) gives
\begin{equation}\phantomsection\label{eq-elast-pi-baseline}{
e^{(1)}(Q,\,x) \;=\; \frac{P'(Q)}{Q\,\bigl[u''(Q,\,x) - P''(Q)\bigr]}.
}\end{equation}

\textbf{Curvature change.} Consider the family
\(\Pi_\sigma(P)(Q) = P(Q)^\sigma\), with reference value
\(\sigma_0 = 1\). This applies a power-law transformation, altering the
\emph{curvature} of the schedule and the rebate structure faced by
high-volume buyers. We have \[
\Pi_\sigma(P(Q))^\prime \;=\; \sigma\, P(Q)^{\sigma - 1}\, P'(Q),
\] and therefore \[
\frac{\partial}{\partial \sigma}\,\Pi_\sigma(P(Q))^\prime \bigg|_{\sigma = 1} \;=\; P'(Q)\,\bigl(1 + \ln P(Q)\bigr).
\] Hence. evaluating Equation~\ref{eq-elast} at \(\sigma_0 = 1\) gives
\begin{equation}\phantomsection\label{eq-elast-gamma-baseline}{
e^{(1)}(Q,\,x) \;=\; \frac{P'(Q)\,\bigl(1 + \ln P(Q)\bigr)}{Q\,\bigl[u''(Q,\,x) - P''(Q)\bigr]}.
}\end{equation}

\begin{remark}
A third natural family is the convex combination $\Pi_\sigma(P)(Q) = \sigma\, P(Q) + (1-\sigma)\, P_0(Q)$, where $P_0$ is known. In this case, $\frac{\partial}{\partial \sigma}\, \Pi_\sigma(P)'(Q) = P'(Q) - P_0'(Q)$, and the elasticity at $\sigma = 1$ is given by
$$
e^{(1)}(Q,\,x) \;=\; \frac{P'(Q) - P_0'(Q)}{Q\,\bigl[u''(Q,\,x) - P''(Q)\bigr]}.
$$
This measures the sensitivity of demand when shifting the schedule toward a target price $P_0$.
\end{remark}

Figure~\ref{fig-figElastCIpi} and Figure~\ref{fig-figElastCIgam} show
the baseline estimates (\(\sigma = 1\)) together with 90\% confidence
bands. We find that the two counterfactuals lead to qualitatively
similar conclusions, with low-volume buyers having higher elasticity
than high-volume ones. High-volume buyers may be purchasing the service
independently of the total price, especially if the mail carrier
operates as a monopolist in the postal market, either because of lack of
competition, or because competitors are unable to handle high-volume
requests. On the contrary, low-volume buyers may decide to pursue other
advertising channels if the service is over-priced.

\begin{figure}[ht]

\begin{minipage}{0.50\linewidth}

\centering{

\includegraphics{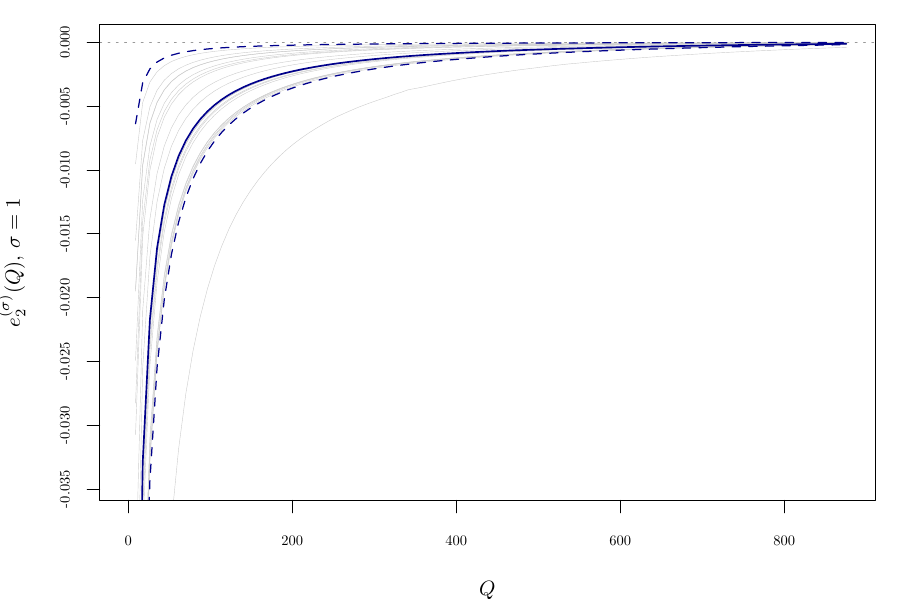}

}

\subcaption{\label{fig-figElastCIpi-1}Breakdown}

\end{minipage}%
\begin{minipage}{0.50\linewidth}

\centering{

\includegraphics{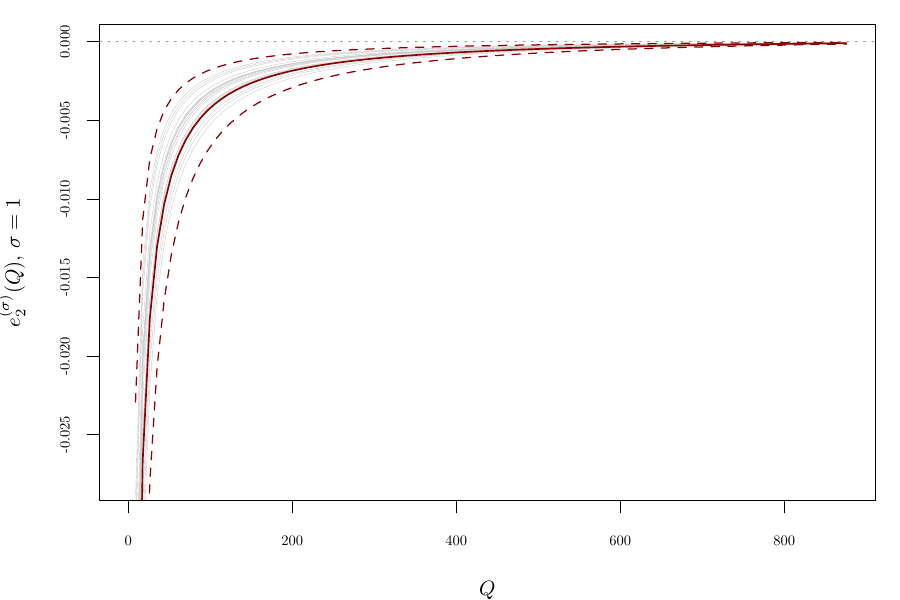}

}

\subcaption{\label{fig-figElastCIpi-2}Car Dealers}

\end{minipage}%
\newline
\begin{minipage}{0.50\linewidth}

\centering{

\includegraphics{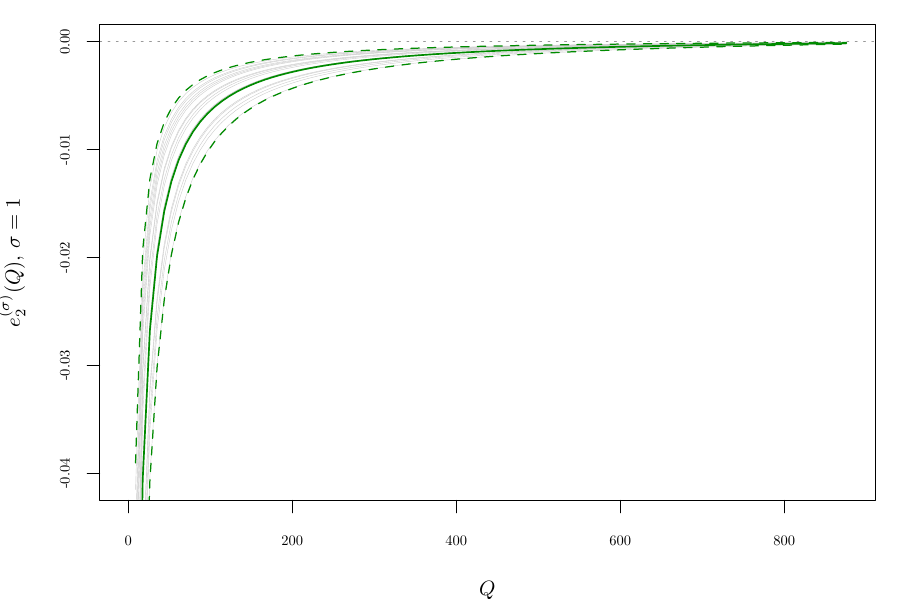}

}

\subcaption{\label{fig-figElastCIpi-3}Elec Shops}

\end{minipage}%
\begin{minipage}{0.50\linewidth}

\centering{

\includegraphics{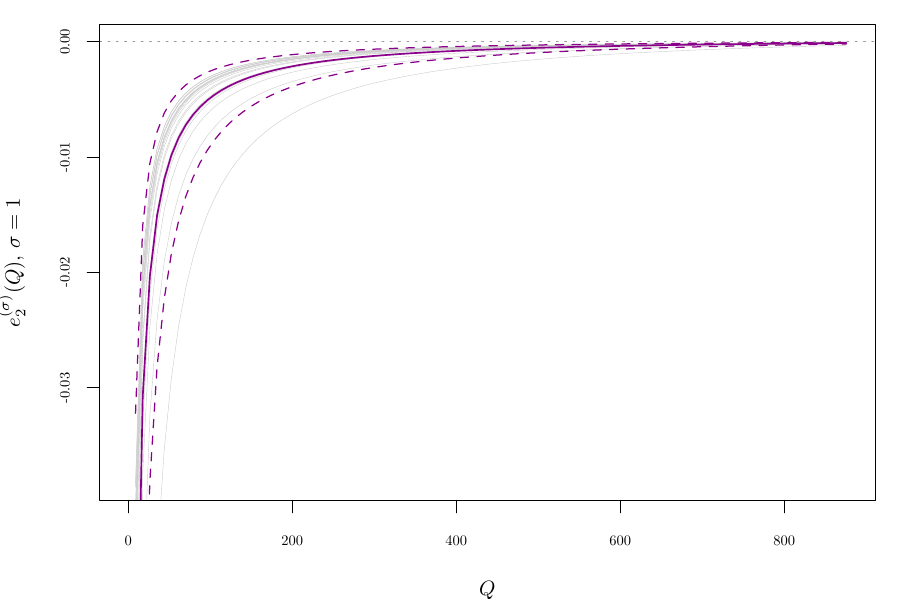}

}

\subcaption{\label{fig-figElastCIpi-4}Health Aid}

\end{minipage}%

\caption{\label{fig-figElastCIpi}Baseline level elasticity
(\(\sigma=1\)) for four sectors with 90\% bootstrap confidence bands
(dashed).}

\end{figure}%

\begin{figure}[ht]

\begin{minipage}{0.50\linewidth}

\centering{

\includegraphics{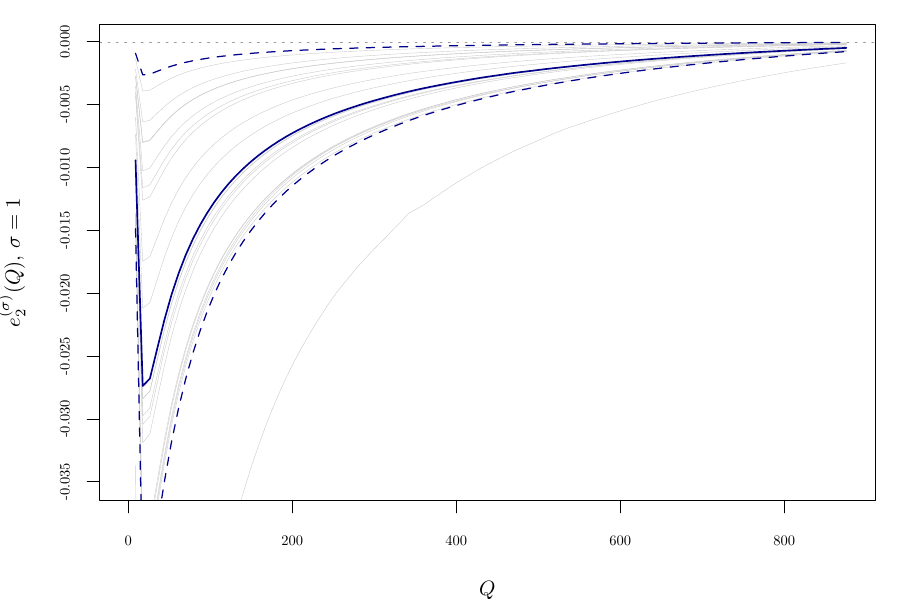}

}

\subcaption{\label{fig-figElastCIgam-1}Breakdown}

\end{minipage}%
\begin{minipage}{0.50\linewidth}

\centering{

\includegraphics{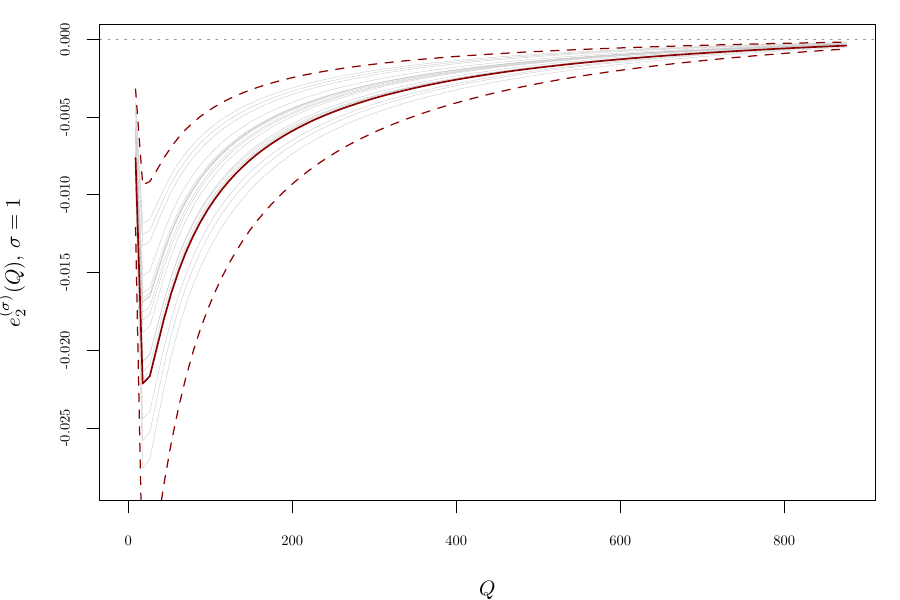}

}

\subcaption{\label{fig-figElastCIgam-2}Car Dealers}

\end{minipage}%
\newline
\begin{minipage}{0.50\linewidth}

\centering{

\includegraphics{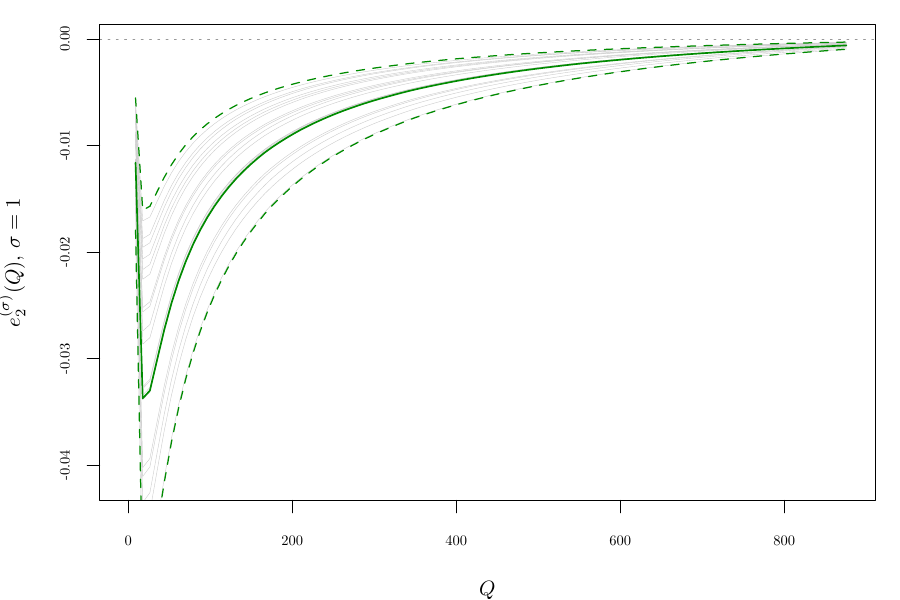}

}

\subcaption{\label{fig-figElastCIgam-3}Elec Shops}

\end{minipage}%
\begin{minipage}{0.50\linewidth}

\centering{

\includegraphics{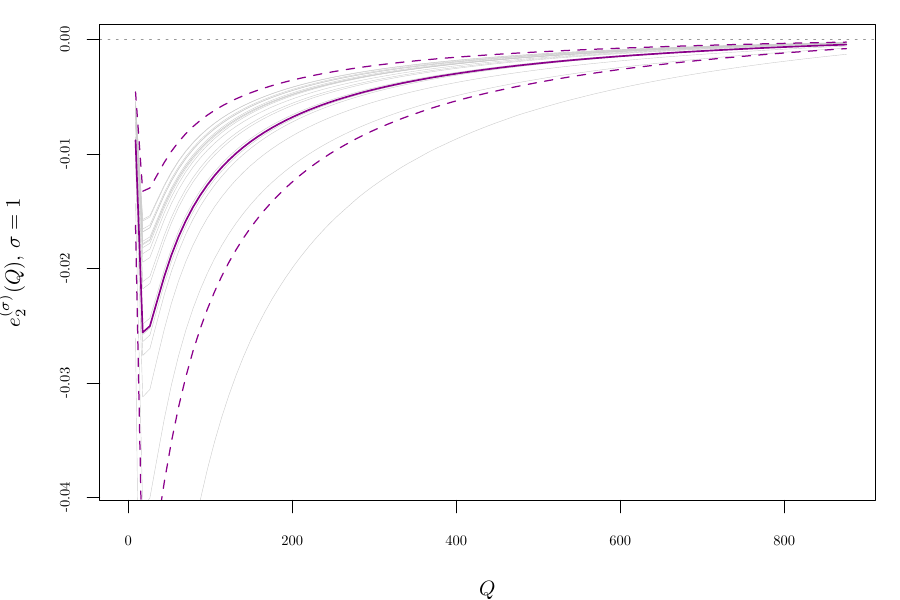}

}

\subcaption{\label{fig-figElastCIgam-4}Health Aid}

\end{minipage}%

\caption{\label{fig-figElastCIgam}Baseline curvature elasticity
(\(\sigma=1\)) for four sectors with 90\% bootstrap confidence bands
(dashed). Bands are omitted when bootstrap results are not yet
available.}

\end{figure}%

\section{Conclusions}\label{sec-conclusion}

This paper develops a nonparametric framework for identifying and
estimating consumer preferences from repeated cross-sections when
consumers face nonlinear pricing schedules. The key insight is that
variation across different price regimes generates a pair of identifying
restrictions whose difference yields a linear iterative functional
equation in the quantile function of unobserved preference
heterogeneity. We show that this equation has a unique solution under
mild regularity conditions, derive a closed-form representation as a
convergent infinite series, and provide tools for estimation and
inference.

From a methodological standpoint, our approach contributes a new
regularization strategy for a class of ill-posed inverse problems
defined by iterative functional equations. We show that the linear
operator defining the problem is compact but does not admit a bounded
inverse, so that the inverse problem is ill-posed. Our iterative
truncation scheme regularizes the inverse by retaining only finitely
many terms of the series solution. Unlike spectral or Tikhonov
regularization, which require choosing a continuous tuning parameter and
inverting an operator, the iterative approach yields a closed-form
expression. The resulting estimator achieves a near-parametric
convergence rate.

The empirical application to unaddressed advertising data from a
European mail carrier illustrates how the framework can be used in
practice. By estimating sector-specific utility functions and preference
distributions, we recover demand elasticities under two types of
counterfactual price experiments: uniform level shifts and curvature
changes. The estimated elasticities reveal that low-volume buyers are
substantially more price-sensitive than high-volume ones, a pattern
consistent across sectors and robust to the type of counterfactual
considered. This finding is consistent with market power on the part of
the mail carrier over large-scale clients, who may lack viable
alternatives for high-volume distribution, while smaller buyers can
substitute toward other advertising channels. These estimates could
inform the design of nonlinear tariffs and the evaluation of regulatory
interventions in postal markets.

Several limitations of the current framework should be acknowledged.
First, our identification strategy requires at least two distinct price
schedules observed on comparable populations; when only a single
schedule is available, the utility function and preference distribution
are not separately identified. Second, we assume that the distribution
of unobserved preference heterogeneity is stable across the two samples.
If preferences shift between periods, for instance, due to aggregate
demand shocks, the estimator is inconsistent. Third, while we extend the
framework to discrete covariates, the case of continuously distributed
observed heterogeneity requires nonparametric smoothing in \(X\) and a
modified asymptotic theory that we do not develop here. Fourth, the
choice of the regularization parameter \(N\) is currently guided by a
heuristic stopping rule; a formal, data-driven selection procedure
remains an open problem. Finally, the empirical application relies on a
parametric specification for the price function at the first stage,
which reintroduces functional form assumptions even though the utility
function is estimated nonparametrically. Extending the framework to
allow for fully nonparametric price estimation is a natural next step.

These limitations suggest several directions for future research. A
formal analysis of the optimal choice of \(N\), possibly via
cross-validation or information criteria adapted to the iterative
structure, would be valuable. Extending the identification and
estimation theory to settings with more than two price schedules could
yield efficiency gains and enable over-identification tests of the
model's restrictions. The framework could also be extended to other
settings where repeated observations under different policy regimes are
used to recover structural objects --- such as labor supply under
changing tax schedules or insurance demand under varying contract menus.

\newpage{}

\section*{References}\label{references}
\addcontentsline{toc}{section}{References}

\phantomsection\label{refs}

\newpage{}

\setcounter{section}{0}
\renewcommand{\thesection}{\Alph{section}}

\setcounter{table}{0}
\renewcommand{\thetable}{A\arabic{table}}

\setcounter{figure}{0}
\renewcommand{\thefigure}{A\arabic{figure}}

\section*{Appendix}\label{appendix}
\addcontentsline{toc}{section}{Appendix}

\section{Proofs}\label{proofs}

\subsection{\texorpdfstring{Proof of
Lemma~\ref{lem-lem:cdfcont}}{Proof of Lemma~}}\label{proof-of-lem-lemcdfcont}

\begin{enumerate}
\def\labelenumi{(\roman{enumi})}
\tightlist
\item
  Let \(G_j = F_\varepsilon  \circ \varphi_j\). That is, \[
  G_j(q) = F_\varepsilon(-\varphi_j(q)).
  \] Because of Assumption \ref{ass:regcond}(iii), \(F_\varepsilon\) is
  monotone and thus differentiable almost everywhere. Hence, \[
  \frac{dG_j}{dq} = -f_\varepsilon(-\varphi_j(q)) \frac{d\varphi_j}{dq}> 0,
  \] where the conclusion follows from Assumption
  \ref{ass:regcond}(iii), and implies that \(G_j\) is strictly
  increasing in \(q\) over its support. Furthermore, because of
  Assumptions \ref{ass:regcond}(ii), \ref{ass:errorbounds}, and
  \ref{ass:priceder}, \(G_j(0) = F_\varepsilon(0) = 0\), and
  \(\underset{\bar{}}{q} = G_j^{-1} \left( 0 \right)  = 0\) and, because
  of the compactness of the support of \(\varepsilon\),
  \(G_j(\bar{q}_j ) = F_\varepsilon(\bar{\varepsilon}) = 1\), which
  implies
  \(\bar{q}_j = G_j^{-1} \left( F_\varepsilon(\bar{\varepsilon}) \right) < \infty\),
  for \(j = 1,2\).
\item
  From Assumption \ref{ass:priceder}, we have that \[
  P(\varepsilon > e \vert Q_j > 0) = P(\varepsilon > e \vert \varphi(Q_j) < 0)  = P(\varepsilon > e \vert \varepsilon > 0),
  \] for \(j = 1,2\). The result of the second part of the Lemma
  immediately follows.
\end{enumerate}

\subsection{\texorpdfstring{Proof of
Proposition~\ref{prp-prp:identification}}{Proof of Proposition~}}\label{proof-of-prp-prpidentification}

Let \[
\mathcal{N}(T) = \lbrace \Lambda : \Lambda (\beta(\alpha)) - \Lambda(\alpha) = 0 \rbrace
\] be the null space of the operator \(T\). If \(\Lambda\) is equal to a
constant then \(\Lambda \in \mathcal{N}(T)\). To show that this is also
a sufficient condition, notice that \(T \Lambda = 0\), for all
\(\alpha\) such that \(\beta(\alpha) = \alpha\). Because of Assumption
\ref{ass:cdfintersect}, this equation only admits a finite number of
solutions \(0 < \alpha_1 < \dots < \alpha_k \leq 1\). Let us assume that
\(\beta(\alpha) > \alpha\) for \(\alpha \in (0,\alpha_1)\). Because of
the continuity of \(\beta(\alpha)\), we then need to have that
\(\beta(\alpha) < \alpha\) for \(\alpha \in (\alpha_1,\alpha_2)\), and
\(\beta(\alpha) > \alpha\) for \(\alpha \in (\alpha_2,\alpha_3)\). More
generally, for \(p \geq 1\), we must have that \[
\begin{cases}
\beta(\alpha) < \alpha  & \text{ for } \alpha \in (\alpha_{2p - 1},\alpha_{2p})\\
\beta(\alpha) > \alpha & \text{ for }  \alpha \in (\alpha_{2p},\alpha_{2p+1}).
\end{cases}
\] Let
\(\beta^{(k)} (\alpha) = \beta\left( \beta^{(k-1)} (\alpha) \right)\)
with \(\beta^{(0)} (\alpha) = \alpha\). As \[
\Lambda (\beta(\alpha)) - \Lambda(\alpha) = 0,
\] we also must have that \[
\Lambda (\beta^{(k)}(\alpha)) - \Lambda(\alpha) = 0.
\] When \(\alpha \in (\alpha_{2p},\alpha_{2p+1})\),
\(\beta^{k}(\alpha) > \beta^{k-1}(\alpha)\), for all \(k \geq 1\), which
implies that \(\beta^{(k)} (\alpha)\) converges to \(\alpha_{2p+1}\) as
\(k \rightarrow \infty\). When
\(\alpha \in (\alpha_{2p - 1},\alpha_{2p})\),
\(\beta^{k}(\alpha) < \beta^{k-1}(\alpha)\), and
\(\beta^{(k)} (\alpha)\) converges to \(\alpha_{2p-1}\), as
\(k \rightarrow \infty\).

The latter implies that
\(\Lambda(\alpha) = \Lambda(\alpha_{2p}), \forall \alpha \in (\alpha_{2p-1},\alpha_{2p+1})\).
That is, \(\Lambda\) is constant on \((\alpha_{2p-1},\alpha_{2p+1})\).
As \(\Lambda\) is continuous by Assumption \ref{ass:regcond}(iii), then
\(\Lambda\) is constant on \([0,1]\). This concludes the proof.

\subsection{\texorpdfstring{Proof of
Lemma~\ref{lem-lembias}}{Proof of Lemma~}}\label{proof-of-lem-lembias}

We directly have that \begin{align*}
\vert \Lambda^{\ast,N}(\alpha) - \Lambda^{\ast}(\alpha) \vert  =& \left\vert \sum_{k = N+1}^\infty r\left( \beta^{(k)} (\alpha) \right) \right\vert \leq  \sum_{k = N+1}^\infty \vert r\left( \beta^{(k)} (\alpha) \right)  \vert  \\
& \quad \leq \sum_{k = N+1}^\infty C \theta^\kappa \vert\beta^{(k-1)} (\alpha) \vert^\kappa \leq C \sum_{k = N+1}^\infty \theta^{k\kappa} \vert \alpha \vert^\kappa \\
=& C \vert \alpha \vert^{\kappa}  \frac{\theta^{(N+1)\kappa}}{1 - \theta^\kappa},
\end{align*} where the second line follows from repeated applications of
Assumption \ref{ass:priceregholder} and the third line follows from
Assumption \ref{ass:attractive}.

This implies that \[
\sup_{\alpha \in (0,1)}\vert \Lambda^{\ast,N}(\alpha) - \Lambda^{\ast}(\alpha) \vert \leq  C \frac{\theta^{(N+1)\kappa}}{1 - \theta^\kappa} \sup_{\alpha \in (0,1)}\vert \alpha \vert^{\kappa} = C \frac{\theta^{(N+1)\kappa}}{1 - \theta^\kappa}.
\] This concludes the proof.

\subsection{\texorpdfstring{Proof of
Proposition~\ref{prp-prp:betakest}}{Proof of Proposition~}}\label{proof-of-prp-prpbetakest}

We can prove the result by induction. For \(k = 1\), and letting
\(Q_\alpha = G_{1}^{-1} (\alpha)\) be the \(\alpha\) percentile of the
distribution of \(Q_1\), we have that \begin{align*}
\hat{\beta}_{h} (\alpha)- \beta (\alpha)  =& \hat{G}_{h,2} \hat{G}_{h,1}^{-1} (\alpha) - G_{2} \hat{G}_{h,1}^{-1} (\alpha)  + G_{2} \hat{G}_{h,1}^{-1} (\alpha) - G_{2} (Q_\alpha)\\
=& \left( \hat{G}_{h,2}  - G_2 \right)\left( \hat{G}_{h,1}^{-1} (\alpha) \right) -  \left( \hat{G}_{h,2}  - G_2 \right)(Q_\alpha) \\
& \quad + \left( \hat{G}_{h,2}  - G_2 \right)(Q_\alpha) + g^\prime_2(Q_\alpha) \left( \hat{G}_{h,1}^{-1} (\alpha)  - Q_\alpha \right) + \frac{g^\prime_2(\tilde{Q}_\alpha)}{2} \left( \hat{G}_{h,1}^{-1} (\alpha)  - Q_\alpha \right)^2\\ 
=& \left( \hat{g}_{h,2}  - g_2 \right) \left( Q_\alpha\right) \left( \hat{G}_{h,1}^{-1} (\alpha) - Q_\alpha \right) + \frac{\left( \hat{g}^\prime_{h,2}  - g^\prime_2 \right) \left( \tilde{Q}_\alpha\right)}{2} \left( \hat{G}_{h,1}^{-1} (\alpha) - Q_\alpha \right)^2 \\
& \quad + \left( \hat{G}_{h,2}  - G_2 \right)(Q_\alpha) + g_2(Q_\alpha) \left( \hat{G}_{h,1}^{-1} (\alpha)  - Q_\alpha \right) + \frac{g^\prime_2(\tilde{Q}_\alpha)}{2} \left( \hat{G}_{h,1}^{-1} (\alpha)  - Q_\alpha \right)^2,
\end{align*} where the Taylor expansion follows from the smoothness
assumptions on the kernel function, and \(\tilde{Q}_\alpha\) is an
intermediate point between \(\hat{G}_{h,1}^{-1} (\alpha)\) and
\(Q_\alpha\).

From standard results in kernel estimation and empirical process theory,
we have that \[
\begin{aligned}
\sup_{0 <h\leq b_n}\sup_{q \in \mathcal{Q}} \vert \hat{g}_{h,2}(q)  - g_2  (q) \vert =& o(1)\\
\sup_{0 <h\leq b_n}\sup_{q \in \mathcal{Q}} \vert \hat{g}^\prime_{h,2}(q)  - g^\prime_2  (q) \vert =& o(1),
\end{aligned}
\] and \[
\begin{aligned}
\sup_{0 <h\leq b_n}\sup_{\alpha \in [0,1]} \vert \hat{G}_{h,1}^{-1} (\alpha)  - Q_\alpha \vert =& O\left(\sqrt{\frac{\log \log n}{n}}\right)\\
\sup_{0 <h\leq b_n}\sup_{q \in \mathcal{Q}} \vert \hat{G}_{h,2}(q)  -G_2  (q) \vert =& O\left(\sqrt{\frac{\log \log n}{n}}\right).
\end{aligned}
\] Therefore, we directly have that \[
\sup_{0 <h\leq b_n}\sup_{\alpha \in [0,1]} \vert \hat{\beta}_{h} (\alpha)- \beta (\alpha) \vert  = O \left( n^{-1/2} \left( \log \log n \right)^{1/2}\right).
\] For \(k = 2\), we have \begin{align*}
\hat{\beta}^{(2)}_{h} (\alpha)- \beta^{(2)} (\alpha)  =& \hat{\beta}^{(2)}_{h} (\alpha) - \beta\left(\hat{\beta}_h(\alpha)\right) + \beta(\hat{\beta}_h(\alpha)) - \beta^{(2)} (\alpha)\\
=& \left( \hat{\beta}_h - \beta \right) \left( \hat{\beta}_h(\alpha)\right) - \left( \hat{\beta}_h - \beta \right) \left( \beta(\alpha)\right) \\
& \quad + \left( \hat{\beta}_h - \beta \right) \left( \beta(\alpha)\right) + \beta\left(\hat{\beta}_h(\alpha)\right) - \beta^{(2)}(\alpha)\\
=& \left( \hat{\beta}^\prime_h(\beta (\alpha))  - \beta^\prime(\beta (\alpha)) \right) \left( \hat{\beta}_{h} (\alpha)- \beta (\alpha)\right) + \frac{1}{2}\left( \hat{\beta}^{\prime \prime}_h(\tilde{\beta}(\alpha))  - \beta^{\prime \prime}(\tilde{\beta}(\alpha)) \right) \left( \hat{\beta}_{h} (\alpha)- \beta (\alpha)\right)^2 \\
& \quad + \left( \hat{\beta}_h - \beta \right) (\beta(\alpha)) + \beta^\prime( \beta(\alpha) ) \left( \hat{\beta}_{h} (\alpha)- \beta (\alpha)\right) + \frac{\beta^{\prime \prime}( \tilde{\beta}(\alpha) )}{2} \left( \hat{\beta}_{h} (\alpha)- \beta (\alpha)\right)^2,
\end{align*} where \[
\begin{aligned}
\beta^\prime(\beta(\alpha)) =& \frac{g_2 \left( G^{-1}_1(\beta(\alpha))\right) }{g_1 \left( G^{-1}_1(\beta(\alpha))\right)} \\
\beta^{\prime\prime}(\beta(\alpha)) =& \frac{g^\prime_2 \left( G^{-1}_1(\beta(\alpha))\right) - \frac{g^\prime_1 \left( G^{-1}_1(\beta(\alpha))\right)}{g_1 \left( G^{-1}_1(\beta(\alpha))\right)} g_2 \left( G^{-1}_1(\beta(\alpha))\right)}{g^2_1 \left( G^{-1}_1(\beta(\alpha))\right)}.
\end{aligned}
\] From the Assumption that the density \(g_1\) is uniformly bounded
away from \(0\) on \(\mathcal{Q}\), we have similar convergence results
for the first and second derivative of \(\hat{\beta}_h\), and also the
uniform boundedness of the second derivative of \(\beta\) wrt to its
argument. Hence, we have that \[
\sup_{0 <h\leq b_n}\sup_{\alpha \in [0,1]} \vert \hat{\beta}^{(2)}_{h} (\alpha)- \beta^{(2)} (\alpha) \vert = \sup_{0 <h\leq b_n}\sup_{\alpha \in [0,1]} \vert \hat{\beta}_h \left( \beta(\alpha) \right) - \beta^{(2)} (\alpha) \vert (1 + o(1)) = O \left( n^{-1/2} \left( \log \log n \right)^{1/2}\right).
\] We can iterate this line of proof for all \(k > 2\), by decomposing
each term as above, and obtaining the Taylor expansion around the point
\(\beta^{(k-1)}(\alpha)\) to obtain the result of the Theorem. This
concludes the proof.

\subsection{\texorpdfstring{Proof of
Lemma~\ref{lem-variance}}{Proof of Lemma~}}\label{proof-of-lem-variance}

Let \(\hat{r} = (P^\prime_1 - P^\prime_2) \circ \hat{G}^{-1}_{h,1}\). We
directly have that \begin{align*}
\vert \hat{\Lambda}_h^{\ast,N}(\alpha)  - \Lambda^{\ast,N}(\alpha) \vert  =& \left\vert \sum_{k = 0}^N  (P^\prime_1 - P^\prime_2) \left( \hat{G}^{-1}_{h,1} - G^{-1}_{1}\right)\left( \hat{\beta}^{(k)}_h (\alpha) \right) + r\left( \hat{\beta}^{(k)}_h (\alpha) \right)  - r\left(\beta^{(k)} (\alpha) \right) \right\vert \\
& \quad  \leq \sum_{k = 0}^N \left\vert  (P^\prime_1 - P^\prime_2) \left( \hat{G}^{-1}_{h,1} - G^{-1}_{1}\right)\left( \hat{\beta}^{(k)}_h (\alpha) \right)  \right\vert +  \sum_{k = 0}^N \left\vert  r\left( \hat{\beta}^{(k)}_h (\alpha) \right)  - r\left(\beta^{(k)} (\alpha) \right)  \right\vert.
\end{align*} where the last line follows from repeated applications of
the triangle inequality.

From Proposition~\ref{prp-prp:cdfest} and
Proposition~\ref{prp-prp:betakest}, we know that \(\hat{G}^{-1}_{h,1}\)
and \(\hat{\beta}^{(k)}_h (\alpha)\) are strongly consistent estimators,
uniformly on \(h\) and for all \(\alpha \in [0,1]\). Therefore the first
term can be bounded in such a way that \[
\sup_{0 <h \leq b_n} \sup_{\alpha \in [0,1]} \left\vert  (P^\prime_1 - P^\prime_2) \left( \hat{G}^{-1}_{h,1} - G^{-1}_{1}\right)\left( \hat{\beta}^{(k)}_h (\alpha) \right)  \right\vert = O \left( n^{-1/2} \left( \log \log n \right)^{1/2} \right), 
\] for all \(k = 0,1,\dots,N\). Since we have \(N+1\) of these terms, we
finally obtain \[
\sum_{k = 0}^N \sup_{0 <h\leq b_n} \sup_{\hat{\beta}^{(k)}_h (\alpha) \in [0,1]}  \left\vert  (P^\prime_1 - P^\prime_2) \left( \hat{G}^{-1}_{h,1} - G^{-1}_{1}\right)\left( \hat{\beta}^{(k)}_h (\alpha) \right)  \right\vert = O \left( (N+1) n^{-1/2} \left( \log \log n \right)^{1/2} \right).
\]

Similarly, for the second term, we can use the mean-value theorem and
the assumption of uniform boundedness of \(r^\prime\) to write
\begin{align*}
\sup_{0 <h\leq b_n} \sup_{\alpha \in [0,1]} & \left\vert  r\left( \hat{\beta}^{(k)}_h (\alpha) \right)  - r\left(\beta^{(k)} (\alpha) \right)  \right\vert =  \sup_{0 <h\leq b_n} \sup_{\alpha \in [0,1]} \left\vert  r^\prime \left( \tilde{\beta}^{(k)}(\alpha) \right) \left( \hat{\beta}^{(k)}_h (\alpha) - \beta^{(k)} (\alpha) \right)  \right\vert \\
& \quad \leq R \sup_{0 <h\leq b_n} \sup_{\alpha \in [0,1]} \left\vert   \hat{\beta}^{(k)}_h (\alpha) - \beta^{(k)} (\alpha) \right\vert = O \left( n^{-1/2} \left( \log \log n \right)^{1/2} \right),
\end{align*} where the final result follows again from
Proposition~\ref{prp-prp:cdfest}. This concludes the proof.

\subsection{\texorpdfstring{Proof of
Lemma~\ref{lem-lem:betaconv}}{Proof of Lemma~}}\label{proof-of-lem-lembetaconv}

We prove the result of the Lemma by induction. Let us start with
\(k = 1\). Then \(\beta(\alpha) = G_2 \left( G^{-1}_1(\alpha)\right)\).
As \(\mathcal{Q}\) is bounded and the pdfs \(g_1\) and \(g_2\) are
bounded away from \(0\) and \(\infty\), we can apply the results about
Hadamard differentiability for empirical processes (see Van der Vaart
(1998), Ch. 19-21). In particular, letting
\(G_{1t} = (1 - t) G_1 + t \delta_1\), and
\(G_{2t} = (1 - t) G_2 + t \delta_2\), where the \(\delta_j\)'s are the
Dirac's measures at the sample points, we have \[
\left. \frac{d}{dt} G_{2t} \left( G^{-1}_{1t}(\alpha)\right) \right\vert_{t = 0} = \left(\delta_2 - G_2 \right) \left( G_1^{-1}(\alpha)\right) - \frac{g_2\left( G_1^{-1}(\alpha)\right)}{g_1\left( G_1^{-1}(\alpha)\right)} \left( \delta_1\left(G_1^{-1}(\alpha)\right) - \alpha \right),
\] from which it follows from the result in
Proposition~\ref{prp-prp:cdfest}, that \[
\sup_{0 <h\leq b_n}  \sqrt{n} \left( \hat{\beta}(\alpha) - \beta(\alpha) \right) \Rightarrow \mathbb{B}_1 \left(  \beta (\alpha) \right) - \mathbb{B}_2(\alpha) \frac{g_2 \left( G_1^{-1} (\alpha) \right) }{g_1 \left( G_1^{-1} (\alpha) \right) }.
\]

For \(k = 2\), let \(\beta^{(2)}(\alpha) = \beta (\beta(\alpha))\), with
\(\beta_t = G_{2t} \circ G^{-1}_{1t}\). Then,

\begin{align*}
\left. \frac{d}{dt} \beta_t^{(2)} (\alpha) \right\vert_{t = 0} =& \left. \frac{d}{dt} G_{2t} \left( G^{-1}_{1t}(\beta_t (\alpha)) \right) \right\vert_{t = 0} \\
=& \left. \frac{d}{dt} G_{2t} \left( G^{-1}_{1t}(\beta(\alpha))\right) \right\vert_{t = 0} + \left. \frac{d}{d\beta_t(\alpha)} G_{2t} \left( G^{-1}_{1t}(\beta_t(\alpha)\right) \frac{d}{dt} \beta_t(\alpha) \right\vert_{t = 0},
\end{align*} which implies

\begin{align*}
\sup_{0 <h\leq b_n}  & \sqrt{n} \left( \hat{\beta}^{(2)}(\alpha) - \beta^{(2)}(\alpha) \right) \\
& \approx  \sup_{0 <h\leq b_n}  \sqrt{n} \left\lbrace \left( \hat{G}_2 \left( \hat{G}^{-1}_1 (\beta(\alpha))\right) - \beta^{(2)}(\alpha) \right) + \frac{g_2 \left( G_1^{-1} (\beta(\alpha)) \right) }{g_1 \left( G_1^{-1} (\beta(\alpha)) \right) } \left( \hat{\beta}(\alpha) - \beta(\alpha) \right) \right\rbrace \\
& \qquad \Rightarrow \mathbb{B}_1 \left( \beta^{(2)}(\alpha) \right) -  \mathbb{B}_2 \left( \beta (\alpha) \right) \frac{g_2 \left( G_1^{-1} (\beta(\alpha)) \right) }{g_1 \left( G_1^{-1} (\beta(\alpha)) \right) } \\
& \qquad +  \mathbb{B}_1 \left( \beta(\alpha) \right) \frac{g_2 \left( G_1^{-1} (\beta(\alpha)) \right) }{g_1 \left( G_1^{-1} (\beta(\alpha)) \right) } - \mathbb{B}_2 (\alpha) \frac{g_2 \left( G_1^{-1} (\beta(\alpha)) \right) }{g_1 \left( G_1^{-1} (\beta(\alpha)) \right) } \frac{g_2 \left( G_1^{-1} (\alpha) \right) }{g_1 \left( G_1^{-1} (\alpha) \right) }.
\end{align*}

For any \(k > 2\), we can use a similar recursive relationship to obtain
that

\begin{align*}
\left. \frac{d}{dt} \beta_t^{(k)} (\alpha) \right\vert_{t = 0} =& \left. \frac{d}{dt} G_{2t} \left( G^{-1}_{1t}(\beta^{(k-1)}(\alpha))\right) \right\vert_{t = 0} + \left. \frac{d}{d\beta^{(k-1)}_t(\alpha)} G_{2t} \left( G^{-1}_{1t}(\beta^{(k-1)}_t(\alpha)\right) \frac{d}{dt} \beta^{(k-1)}_t(\alpha) \right\vert_{t = 0} \\
=& \left. \frac{d}{dt} G_{2t} \left( G^{-1}_{1t}(\beta^{(k-1)}(\alpha))\right) \right\vert_{t = 0} \\
& \qquad + \left. \frac{d}{d\beta^{(k-1)}_t(\alpha)} G_{2t} \left( G^{-1}_{1t}(\beta^{(k-1)}_t(\alpha)\right) \right\vert_{t = 0} \left[ \left. \frac{d}{dt} G_{2t} \left( G^{-1}_{1t}(\beta^{(k-2)}(\alpha))\right) \right\vert_{t = 0} \right.\\
& \qquad \left. + \left. \frac{d}{d\beta^{(k-2)}_t(\alpha)} G_{2t} \left( G^{-1}_{1t}(\beta^{(k-2)}_t(\alpha)\right) \frac{d}{dt} \beta^{(k-2)}_t(\alpha) \right\vert_{t = 0} \right]. 
\end{align*}

The result of the Lemma follows.

\subsection{\texorpdfstring{Proof of
Theorem~\ref{thm-thm:lambdaconv}}{Proof of Theorem~}}\label{proof-of-thm-thmlambdaconv}

We have, \begin{align*}
\sqrt{n} \left( \hat{\Lambda}_N (\alpha) - \Lambda_N(\alpha) \right) =& \sqrt{n} \sum_{k = 1}^N \left( r\left(\hat{\beta}^{(k)}(\alpha)\right) - r\left(\beta^{(k)}(\alpha)\right)\right)\\
=& \sqrt{n} \sum_{k = 1}^N r^\prime \left(\beta^{(k)}(\alpha)\right)\left( \hat{\beta}^{(k)}(\alpha)- \beta^{(k)}(\alpha) \right)( 1+ o_P(1)).
\end{align*} The result follows by direct application of
Proposition~\ref{prp-prp:cdfest} and Lemma~\ref{lem-lem:betaconv}, and
the delta method for functionals.

\section{Additional Technical Proofs}\label{additional-technical-proofs}

In this section, we provide additional technical proofs about the
properties of the operator in \(T\) in Equation~\ref{eq-eq:functional2}.
In particular, we wish to show that the operator \(T\) does not have a
bounded inverse, and that it is compact. Under the first property, the
inverse problem in Equation~\ref{eq-eq:functional2} is ill-posed. Under
the second property, the singular values of \(T\) are real.

Recall that \(\beta: [0,1] \rightarrow [0,1]\) is a bijection with
\(\beta(0) = 0\) and \(\beta(1) = 1\), which is monotone and thus
differentiable by the Lebesgue theorem. This assumption implies that
\(\mathcal{Q}_1\) and \(\mathcal{Q}_2\) are equal (the two distributions
have the same support, or we consider only the intersection of the two
supports). For simplicity, and without loss of generality, we let
\(\mathcal{E}= [0,1]\). The condition of compact support for
\(\varepsilon\) implies that \(\Lambda\) is uniformly bounded. Finally,
let \(\mathcal{L}\) denote the space of uniformly bounded continuous
functions on \([0,1]\). The operator
\(T:\mathcal{L} \Rightarrow \mathcal{L}\) is defined as

\[
T \Lambda = \Lambda \left( \beta(\alpha) \right) - \Lambda (\alpha).
\]

We can endow \(\mathcal{L}\) with the \(\mathbb{L}^p\) topology for
\(p\geq 1\), so that \(\mathcal{L}\) is a Banach space.

\begin{lemma}[]\protect\hypertarget{lem-lem:opercont}{}\label{lem-lem:opercont}

\(T:\mathcal{L} \Rightarrow \mathcal{L}\) is continuous, for all
\(p \geq 1\).

\end{lemma}

\begin{proof}
For \(1\leq p < \infty\)

\[
\begin{aligned}
\sup_{\Vert \Lambda \Vert_p \leq 1} &\left[ \int \vert \Lambda \left( \beta(\alpha) \right) - \Lambda (\alpha) \vert^p \right]^{\frac{1}{p}} \\
& \quad \leq \sup_{\Vert \Lambda \Vert_p \leq 1} \Vert \Lambda  \left( \beta(\alpha) \right) \Vert_p  + \sup_{\Vert \Lambda \Vert_p \leq 1} \Vert \Lambda  \left( \alpha \right) \Vert_p \leq 1 + M^{1/p},
\end{aligned}
\] where the second line follows from the Minkowski inequality. The
second term is bounded by \(\Vert \Lambda \Vert_p \leq 1\). For the
first term, a change of variables \(\gamma = \beta(\alpha)\) gives

\[
\Vert \Lambda  \left( \beta(\alpha) \right) \Vert^p_p = \int \vert \Lambda( \gamma ) \vert^p \left( \beta^{-1} \right)^\prime (\gamma) d\gamma \leq M \Vert \Lambda \Vert^p_p,
\] where
\(M = \sup_\gamma \left( \beta^{-1} \right)^\prime (\gamma) < \infty\),
since \(\beta\) is a smooth bijection on \([0,1]\) with \(\beta^\prime\)
bounded away from zero. Hence
\(\Vert \Lambda(\beta(\cdot))\Vert_p \leq M^{1/p} \Vert \Lambda \Vert_p\),
and the operator norm of \(T\) is at most \(1 + M^{1/p} < \infty\).

For \(p =\infty\), we use the triangle inequality in a way that

\[
\sup_{\alpha } \vert \Lambda \left( \beta(\alpha) \right) - \Lambda (\alpha) \vert \leq \sup_{\alpha } \vert \Lambda \left( \beta(\alpha) \right)  \vert + \sup_{\alpha } \vert \Lambda (\alpha) \vert \leq 2.
\]
\end{proof}

\begin{lemma}[]\protect\hypertarget{lem-lem:operadjoint}{}\label{lem-lem:operadjoint}

\(T^\ast :\mathcal{L} \Rightarrow \mathcal{L}\) the adjoint operator is
given by

\[
\left(T^\ast \psi\right) (\alpha) = \psi(\beta^{-1}(\alpha )) \left( \beta^{-1} \right)^\prime (\alpha) - \psi(\alpha)
\]

\end{lemma}

\begin{proof}
Let \((\Lambda,\psi) \in \mathcal{L}\). Then,

\[
\begin{aligned}
\langle T\Lambda, \psi \rangle =& \int_{0}^1 \left( \Lambda (\beta(\alpha)) - \Lambda (\alpha) \right) \psi(\alpha) d\alpha \\
& \quad = \int_{0}^1 \Lambda (\gamma) \psi  \left( \beta^{-1}(\gamma ) \right) \left( \beta^{-1}\right)^\prime (\gamma) d\gamma - \int_{0}^1  \Lambda (\alpha) \psi (\alpha) d\alpha \\
& \quad = \int_{0}^1 \Lambda(x) \left(  \psi  \left( \beta^{-1}(\gamma ) \right) \left( \beta^{-1}\right)^\prime (\gamma) - \psi (\gamma) \right) d\gamma = \langle \Lambda, T^\ast\psi \rangle. 
\end{aligned}
\]

This concludes the proof.
\end{proof}

We wish now to analyze the spectrum of the operator \(T^\ast T\). To do
so, we need to solve the functional equation

\[
\left( T^\ast T \phi \right) (\alpha) = \lambda \phi(\alpha).
\]

We have the following.

\begin{lemma}[]\protect\hypertarget{lem-lem:cdfspectrum}{}\label{lem-lem:cdfspectrum}

The operator \(T:\mathcal{L} \Rightarrow \mathcal{L}\) is compact.
Consequently, the eigenvalues \(\lambda_j\), \(j = 1,2,\dots\), of
\(T^\ast T\) satisfy \(\lambda_j \rightarrow 0\) as
\(j \rightarrow \infty\), and \(T\Lambda = r\) is an ill-posed inverse
problem.

\end{lemma}

\begin{proof}
We show compactness of \(T\) via the Arzelà--Ascoli theorem. Let
\(S = \lbrace \Lambda \in \mathcal{L} : \Vert \Lambda \Vert_\infty \leq 1 \rbrace\).
The image \(T(S)\) is uniformly bounded since
\(\Vert T\Lambda \Vert_\infty \leq 2\). For equicontinuity, note that
\(\beta\) is Lipschitz on \([0,1]\) (since \(\beta^\prime\) is bounded),
say with constant \(L\). Then for any \(\Lambda \in S\), \[
\vert T\Lambda(\alpha_1) - T\Lambda(\alpha_2) \vert \leq \vert \Lambda(\beta(\alpha_1)) - \Lambda(\beta(\alpha_2)) \vert + \vert \Lambda(\alpha_1) - \Lambda(\alpha_2) \vert.
\] As \(\Lambda\) ranges over equicontinuous families (which holds since
\(\Lambda \in \mathcal{C}[0,1]\) and \(\beta\) is Lipschitz), \(T(S)\)
is equicontinuous. By Arzelà--Ascoli, \(T(S)\) is relatively compact, so
\(T\) is a compact operator. By the spectral theorem for compact
self-adjoint operators, the eigenvalues of \(T^\ast T\) accumulate at
zero.
\end{proof}

\section{Simulations with additional covariates and endogenous
prices}\label{simulations-with-additional-covariates-and-endogenous-prices}

We consider a design similar to Design 1 in
Section~\ref{sec-montecarlo}. The random variable \(X\) takes values in
\(\mathcal{X} = \lbrace 0,1,2 \rbrace\) with equal probability in Sample
1, and with probabilities equal to \(\lbrace 5/12, 1/3, 1/4 \rbrace\),
respectively, in Sample 2. The conditional distribution of
\(\varepsilon\) given \(X = x\) satisfies \[
F_{\varepsilon \vert X}(\varepsilon \vert X = x) = \varepsilon^{3/(2+x)}, \quad \varepsilon \in [0,\, 1],
\] \(\eta \sim N(0,0.05^2)\), and \(\eta \upmodels (X,\varepsilon)\).
The utility function is \[
u(q,\, x) = 2q - \frac{q^2}{2} - \frac{q^{2+x}}{2+x},
\] so \(u'(q,x) = 2 - q - q^{1+x} \geq 0\) for all \(q \in [0,1)\) and
\(x \in \mathcal{X}\), with \(x=0\) recovering the baseline of
Design\textasciitilde1. The price functions are the same across all
sectors:

\begin{align*}
P_1 =& 2Q - \frac{Q^2}{2} + Q \eta \\
P_2 =& 2Q - \frac{Q^3}{3} + Q \eta,
\end{align*}

with \(\tau_1 = \tau_2 = 1\). Since
\(P_1^\prime(0) = P_2^\prime(0) = 2 = u^\prime(0;x)\), no rescaling is
needed. The structural functions are \(\Lambda_1(Q;x) = Q^{1+x}\) and
\(\Lambda_2(Q;x) = Q(1 + Q^x - Q)\). Inverting the first-order
conditions yields the equilibrium quantity in Period 1: \[
Q^\ast_1 = \varepsilon^{1/(1+X)},
\] and \(Q^\ast_2\) is given by \(1 - \sqrt{1-\varepsilon}\) for
\(X = 0\), by \(\varepsilon\) for \(X = 1\), and by the unique real root
of \(q^3 - q^2 + q = \varepsilon\) for \(X = 2\). Both \(Q^\ast_1\) and
\(Q^\ast_2\) lie in \((0,1)\) for all \(x \in \mathcal{X}\) and
\(\varepsilon \in (0,1)\). The true structural function (Period 2) is
\(\Lambda(Q;x) = Q(1 + Q^x - Q)\).

As \(Q_2^\ast > Q_1^\ast\) almost surely for all \(x \in \mathcal{X}\),
\(G_{2\vert X}\) stochastically dominates \(G_{1\vert X}\) over
\(\mathcal{Q} = [0,1]\) for all \(x \in \mathcal{X}\). This simplifies
the implementation of our estimator in this simulation exercise. We look
at the small-sample properties of the estimated inverse demand and
utility functions, and of the cdf of \(\varepsilon\). Sample sizes are
fixed to \(n = \lbrace 500,1000,2500 \rbrace\).

We estimate the parameters of the price function as outlined at the end
of Section~\ref{sec-priceest}, where the order of the polynomial is
taken to be \(2\) in \(Q\) for both periods. With \(X\) having three
categories, the IV model is just-identified. Results on the
finite-sample performance of the first-step IV estimator are reported in
Table~\ref{tbl-simresults_x_iv}.

\begin{table}

\caption{\label{tbl-simresults_x_iv}Mean and Standard Deviation of the
first-step IV estimator}

\centering{

\begin{tabular}{lrrrrrrrrr}
\toprule
\multicolumn{1}{c}{} & \multicolumn{3}{c}{n = 500} & \multicolumn{3}{c}{n = 1000} & \multicolumn{3}{c}{n = 2500} \\
\cmidrule(l{3pt}r{3pt}){2-4} \cmidrule(l{3pt}r{3pt}){5-7} \cmidrule(l{3pt}r{3pt}){8-10}
  & $RMSE$ & $Bias$ & $StDev$ & $RMSE$ & $Bias$ & $StDev$ & $RMSE$ & $Bias$ & $StDev$\\
\midrule
$\hat{\theta}_{10}$ & 0.0182 & -0.0101 & 0.0152 & 0.0129 & -0.0071 & 0.0108 & 0.0078 & -0.0044 & 0.0065\\
$\hat{\theta}_{11}$ & 0.0280 & 0.0152 & 0.0235 & 0.0198 & 0.0107 & 0.0167 & 0.0120 & 0.0065 & 0.0101\\
$\hat{\theta}_{12}$ & 0.0000 & 0.0000 & 0.0000 & 0.0000 & 0.0000 & 0.0000 & 0.0000 & 0.0000 & 0.0000\\
$\hat{\theta}_{20}$ & 0.0077 & -0.0042 & 0.0065 & 0.0056 & -0.0030 & 0.0047 & 0.0035 & -0.0019 & 0.0029\\
$\hat{\theta}_{21}$ & 0.0000 & 0.0000 & 0.0000 & 0.0000 & 0.0000 & 0.0000 & 0.0000 & 0.0000 & 0.0000\\
\addlinespace
$\hat{\theta}_{22}$ & 0.0254 & 0.0151 & 0.0204 & 0.0184 & 0.0109 & 0.0148 & 0.0115 & 0.0067 & 0.0093\\
\bottomrule
\end{tabular}

}

\end{table}%

We proceed by estimating the pdf of \(\eta\) using residuals from the IV
regression, which is then useful for the estimation of the cdf of
\(\varepsilon\). The pdf is estimated using a second-order Gaussian
kernel with cross-validated bandwidth parameter (see, e.g., Li and
Racine (2007)).

Also, we estimate
\(\hat{r} (\alpha,x) = \Psi^\prime_2 \left(G_{2\vert X}^{-1}(\alpha \vert X = x)\right) \left( \hat\theta_2 - \hat\theta_1 \right)\).
For \(x \in \mathcal{X}\), we then proceed as follows.

\begin{enumerate}
\def\labelenumi{\arabic{enumi}.}
\item
  We estimate the conditional distributions of \(Q_1\) and \(Q_2\) given
  \(X = x\), denoted \(G_{1\vert X}(\cdot \vert x)\), and
  \(G_{2 \vert X}(\cdot \vert x)\), by selecting the observations for
  which \(X = x\), and using a kernel smoothed estimator.
\item
  We estimate \(\beta(\alpha,x)\) as \[
  \hat{\beta}^{(1)}_h (\alpha,x) = \left( \hat{G}_{h,2 \vert X} \right)^{-1} \left( \hat{G}_{h,1\vert X} (\hat{Q}_\alpha\vert X = x), x\right), 
  \] and then compute \[
  \hat{\Lambda}_\zeta^{\ast,1}(\alpha, x) = -\hat{r}(\alpha,x) - \hat{r}\left( \hat{\beta}^{(1)}_h (\alpha,x) \right).
  \]
\item
  Iterate the approach for \(k = 2,\dots,N\), in a way that \[
  \hat{\beta}^{(k)}_h (\alpha,x) = \left( \hat{G}_{h,2 \vert X} \right)^{-1} \left( \hat{G}_{h,1 \vert X} (\hat{\beta}^{(k-1)}_h (\alpha,x),x) \right), 
  \] and \[
  \hat{\Lambda}_\zeta^{\ast,k}(\alpha,x) = -\hat{r}(\alpha,x) - \sum_{j = 1}^{k} \hat{r}\left( \hat{\beta}^{(j)}_h (\alpha,x) \right).
  \] The number of iterations, \(N\), is chosen as before, stopping when
  the variance of \(\hat{\beta}^{(k)}_h (\alpha,x)\) is smaller than a
  tolerance level equal to \(10^{-10}\) and capping the number of
  iterations at \(5 \log(n)\).
\item
  Once we obtain an estimator of \(\Lambda_{\zeta \vert X}\) at \(x\),
  the utility function is obtained as before, replacing the quantile
  function of \(\varepsilon\) with the quantile function of \(\zeta\).
\item
  Finally, an estimator of the cdf of \(\varepsilon\) conditional on
  \(X = x\) is computed as detailed in
  Equation~\ref{eq-fepscond_finalest}.
\end{enumerate}

Bootstrap confidence intervals are obtained using the strategy outlined
in Section~\ref{sec-bootstrap}. Results for one randomly selected sample
with \(n = 1000\) are given in Figure~\ref{fig-figLambda_x}.

\begin{figure}[p]

\begin{minipage}{0.50\linewidth}

\centering{

\includegraphics{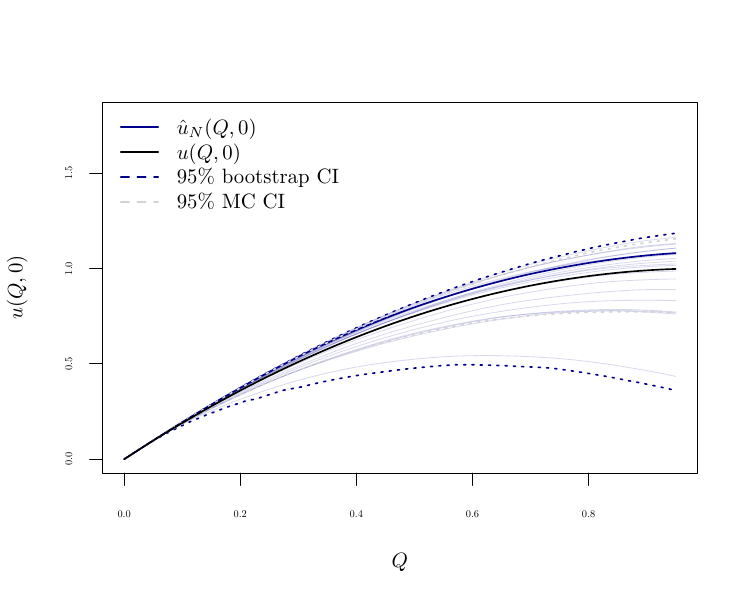}

}

\subcaption{\label{fig-figLambda_x-1}Estimation of \(u\)}

\end{minipage}%
\begin{minipage}{0.50\linewidth}

\centering{

\includegraphics{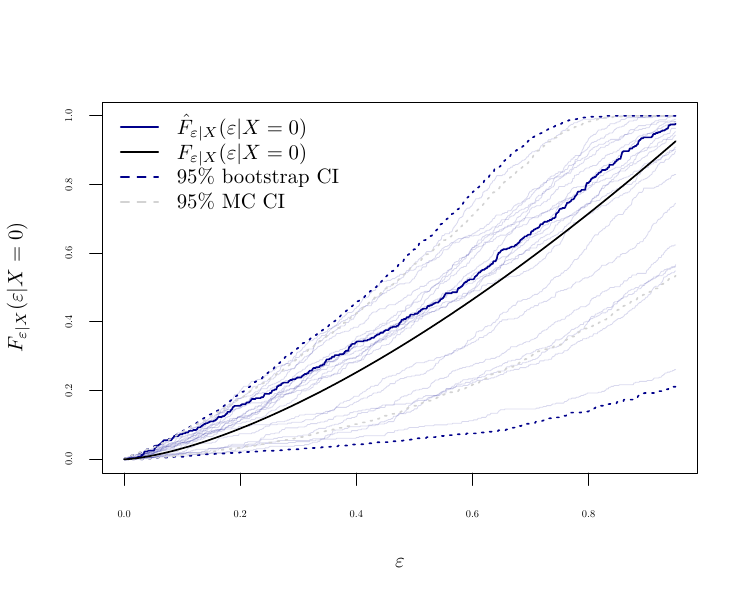}

}

\subcaption{\label{fig-figLambda_x-2}Estimation of
\(F_{\varepsilon \vert X}\)}

\end{minipage}%
\newline
\begin{minipage}{0.50\linewidth}

\centering{

\includegraphics{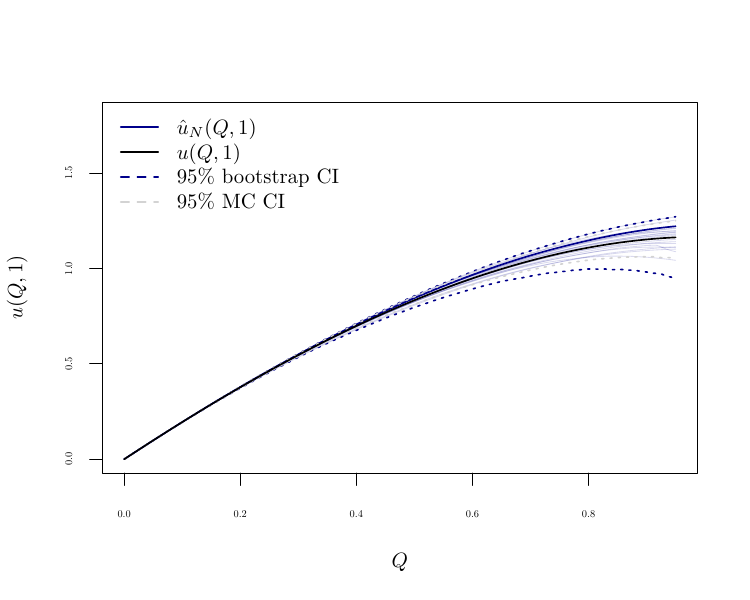}

}

\subcaption{\label{fig-figLambda_x-3}Estimation of \(u\)}

\end{minipage}%
\begin{minipage}{0.50\linewidth}

\centering{

\includegraphics{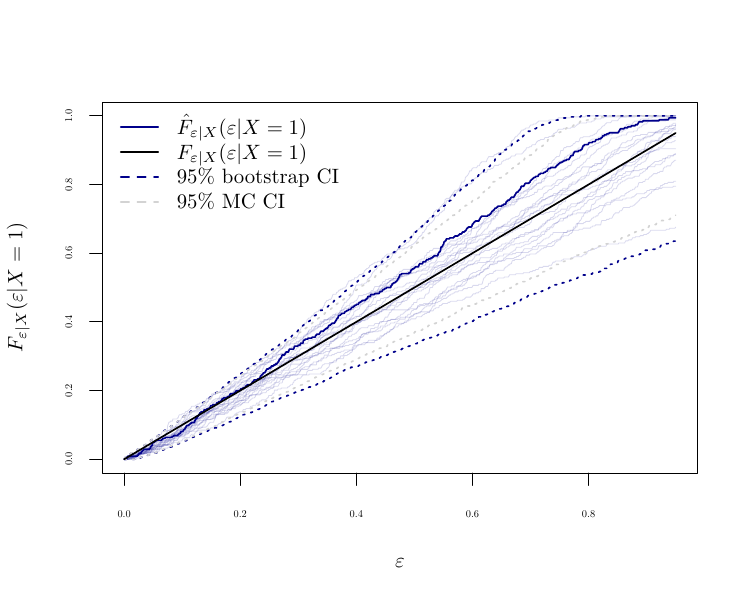}

}

\subcaption{\label{fig-figLambda_x-4}Estimation of
\(F_{\varepsilon \vert X}\)}

\end{minipage}%
\newline
\begin{minipage}{0.50\linewidth}

\centering{

\includegraphics{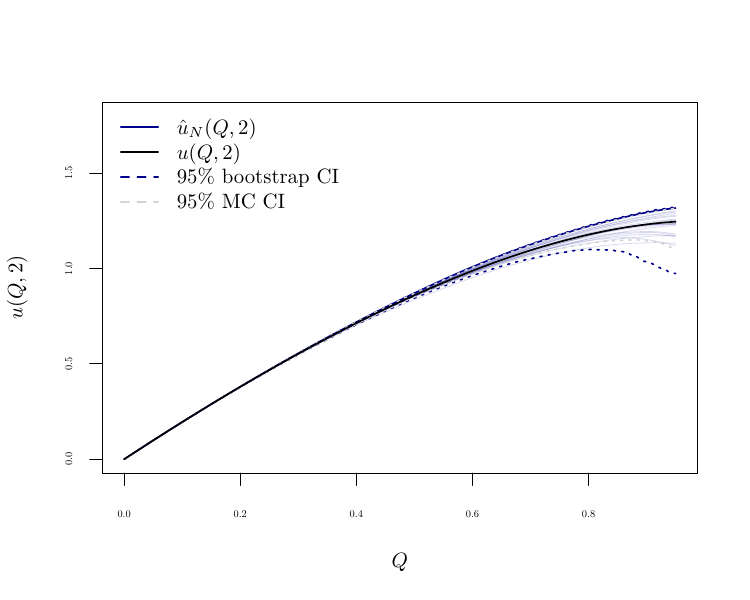}

}

\subcaption{\label{fig-figLambda_x-5}Estimation of \(u\)}

\end{minipage}%
\begin{minipage}{0.50\linewidth}

\centering{

\includegraphics{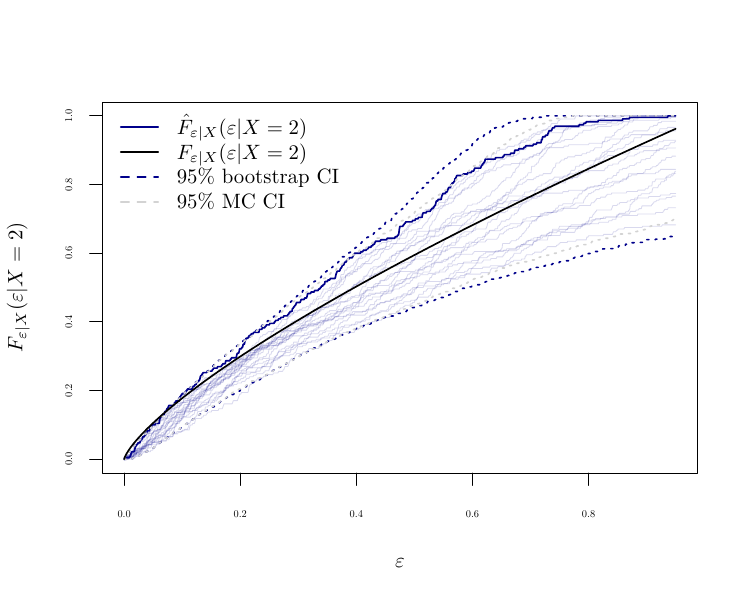}

}

\subcaption{\label{fig-figLambda_x-6}Estimation of
\(F_{\varepsilon \vert X}\)}

\end{minipage}%

\caption{\label{fig-figLambda_x}Estimation of the utility function
(left) and cdf of \(\varepsilon\) (right) for each value of \(X\) and a
randomly selected sample of size \(n = 1000\).}

\end{figure}%

Table~\ref{tbl-simresults_x} gives the average maximum absolute error of
the estimators of the utility function and the cdf of \(\varepsilon\)
conditional on \(X\) and for all sample sizes.

\begin{table}

\caption{\label{tbl-simresults_x}Average maximum absolute deviation of
the estimator conditional on X}

\centering{

\begin{tabular}{lrrrrrr}
\toprule
\multicolumn{1}{c}{} & \multicolumn{3}{c}{$\hat{F}_{\varepsilon \vert X}( \varepsilon \vert X = x)$} & \multicolumn{3}{c}{$\hat{u}_{N}(Q,x)$} \\
\cmidrule(l{3pt}r{3pt}){2-4} \cmidrule(l{3pt}r{3pt}){5-7}
  & $X = 0$ & $X = 1$ & $X = 2$ & $X = 0$ & $X = 1$ & $X = 2$\\
\midrule
n = 500 & 0.2381 & 0.1590 & 0.1601 & 0.1488 & 0.0649 & 0.0633\\
n = 1000 & 0.1581 & 0.1149 & 0.1254 & 0.0822 & 0.0399 & 0.0369\\
n = 2500 & 0.1061 & 0.0741 & 0.0881 & 0.0539 & 0.0239 & 0.0203\\
\bottomrule
\end{tabular}

}

\end{table}%

\end{document}